\documentclass[12pt]{article}
\usepackage{geometry}
\usepackage{IEEEtrantools}
\usepackage{amsmath,amsfonts}
\usepackage{algorithmic}
\usepackage{algorithm}
\usepackage{array}
\usepackage[caption=false,font=normalsize,labelfont=sf,textfont=sf]{subfig}
\usepackage{textcomp}
\usepackage{stfloats}
\usepackage{url}
\usepackage{verbatim}
\usepackage{graphicx}
\usepackage{cite}
\usepackage{macros}
\usepackage{tikz}
\usepackage{float}
\usepackage{setspace}
\usetikzlibrary{positioning}
\usetikzlibrary{calc,matrix} 
\usepackage{caption}
\usepackage{enumerate}
\newtheorem{theorem}{Theorem}
\newtheorem{lemma}{Lemma}
\newtheorem{remark}{Remark}

  \newtheorem{assumption}{Assumption}
 
\begin{document}

\onehalfspacing

\title{Rates and Nearly-Lossless State over a 
  Multi-Access Channel with Causal State Information}

\author{Amos Lapidoth \and Baohua Ni%
    \thanks{A. Lapidoth and B. Ni are with the Department of Information Technology and Electrical Engineering, ETH Zurich, Zurich, Switzerland. Email: \{lapidoth, baohni\}@isi.ee.ethz.ch}
}

\maketitle

\begin{abstract}
  The problem of transmitting both data and state over a memoryless
  state-dependent multi-access channel with causal state information
  is addressed. The state comprises three independent components: a
  common component, which is revealed to both encoders, and two private
  components, each of which is revealed to a different encoder. The
  data are to be recovered with arbitrarily-small probability of
  error, and the state sequence with arbitrarily-small symbol-error
  rate.

  Under the assumption that the capacity region in the absence of a
  state reconstruction constraint is nonempty, the capacity region is
  derived and is shown to be achievable using a Block-Markov coding
  scheme employing binning and superposition coding.
\end{abstract}

\section{Introduction}

This paper studies the communication rates that can be supported by a
state-dependent discrete memoryless multi-access channel (SD-MAC)
when, in addition to recovering the transmitted data reliably, the
receiver must also recover the state sequence.
We assume that the state is revealed to the transmitters causally and
that the receiver must reconstruct the state sequence
``almost losslessly,'' i.e., that the error rate in its reconstruction
of the state sequence must vanish with the blocklength.

We do not address the case where some distortion is allowed in the
state reconstruction.
Our work can thus be viewed as an extension of
\cite{choudhuri2013causal}
to multi-access networks, albeit for the special case of zero
distortion.
We also do not consider feedback, which is studied in the single-user
case (with distortion constraints) in \cite{bross2018rate}. Finally,
we only consider memoryless SD-MACs where the state is IID. The
general case is beyond the scope of this paper, but see
\cite{lapidoth2025rate} for the Gilbert-Elliott single-user channel.

To address different architectures, we model the state $S$ as
comprising three independent components: a component $S_{0}$ that is
revealed to both encoders, a component $S_{1}$ that is revealed only
to Encoder~1 and a component $S_{2}$ that is revealed only to
Encoder~2; see Figure~\ref{fig:simplified_mac}. Special cases are the
``common state'' architecture (where the multi-access is governed by a
single state sequence that is revealed to both encoders) and the
``independent states'' architecture (where the channel is governed by
two independent states each of which is revealed to a different
encoder).

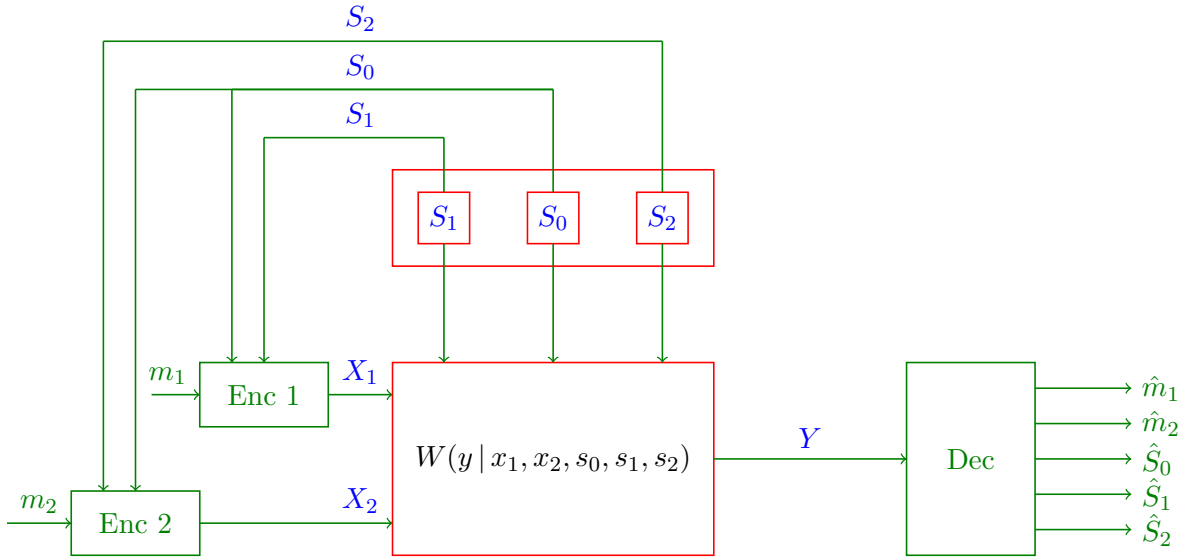
\begin{figure}[htbp]
\definecolor{ao}{rgb}{0,0.5,0}
\small
\begin{tikzpicture}[x=0.85cm, y=0.85cm, line width=0.6pt]

\draw[ao] (-6,2) rectangle (-4,3);
\draw[ao] (-8,0) rectangle (-6,1); 
\node [ao] at(-5,2.5) {Enc 1};
\node [ao] at(-7,0.5) {Enc 2};

\draw[->,ao] (-6.75,2.5) -- (-6,2.5);
\draw[->,ao] (-9,0.5) -- (-8,0.5);
\node [anchor=south,ao] at (-6.5,2.5) {$m_1$};
\node [anchor=south,ao] at (-8.5,0.5) {$m_2$};

\draw[->,ao] (-4,2.5) -- (-3,2.5);
\draw[->,ao] (-6,0.5) -- (-3,0.5);
\node [anchor=south,blue] at (-3.5,2.5) {$X_1$};
\node [anchor=south,blue] at (-3.5,0.5) {$X_2$};

\draw[red] (-3,0) rectangle (2,3);
\node at (-0.5,1.5) {$W(y\,|\,x_1,x_2,s_0,s_1,s_2)$};

\draw[ao,->] (2,1.5) -- (5,1.5);
\node [anchor=south,blue] at (3.5,1.5) {$Y$};

\draw [red] (-3,4.5) rectangle (2,6);
\draw [red] (-2.6,4.85) rectangle (-1.8,5.65);
\draw [red] (-0.9,4.85) rectangle (-0.1,5.65);
\draw [red] (0.8,4.85) rectangle (1.6,5.65);
\node [blue] at (-2.2,5.25) {$S_1$};
\node [blue] at (-0.5,5.25) {$S_0$};
\node [blue] at (1.2,5.25) {$S_2$};

\draw [->,ao] (-2.2,4.85) -- (-2.2,3);
\draw [->,ao] (-0.5,4.85) -- (-0.5,3);
\draw [->,ao] (1.2,4.85) -- (1.2,3);

\draw [ao] (-2.2,5.65) -- (-2.2,6.5);
\draw [ao] (-2.2,6.5) -- (-5,6.5);
\draw [->,ao] (-5,6.5) -- (-5,3);
\node [blue,anchor=south] at (-3.5,6.5) {$S_1$};

\draw [ao] (-0.5,5.65) -- (-0.5,7.25);
\draw [ao] (-0.5,7.25) -- (-5.5,7.25);
\draw [->,ao] (-5.5,7.25) -- (-5.5,3);
\draw [ao] (-0.5,7.25) -- (-7,7.25);
\draw [->,ao] (-7,7.25) -- (-7,1);
\node [blue,anchor=south] at (-3.5,7.25) {$S_0$};

\draw [ao] (1.2,5.65) -- (1.2,8);
\draw [ao] (1.2,8) -- (-7.5,8);
\draw [->,ao] (-7.5,8) -- (-7.5,1);
\node [blue,anchor=south] at (-3.5,8) {$S_2$};

\draw[ao] (5,0) rectangle (7,3);
\node[ao] at(6,1.5) {Dec};

\draw[ao,->] (7,2.6) -- (8.5,2.6);
\draw[ao,->] (7,2.05) -- (8.5,2.05);
\draw[ao,->] (7,1.5) -- (8.5,1.5);
\draw[ao,->] (7,0.95) -- (8.5,0.95);
\draw[ao,->] (7,0.4) -- (8.5,0.4);
\node[ao,anchor=west] at(8.5,2.6) {$\hat{m}_1$};
\node[ao,anchor=west] at(8.5,2.05) {$\hat{m}_2$};
\node[ao,anchor=west] at(8.5,1.5) {$\hat{S}_0$};
\node[ao,anchor=west] at(8.5,0.95) {$\hat{S}_1$};
\node[ao,anchor=west] at(8.5,0.4) {$\hat{S}_2$};

\end{tikzpicture}
\centering
\captionsetup{justification=centering}
\caption{The network under consideration.}
\label{fig:simplified_mac}
\end{figure}

Our main result is a characterization of the capacity region, i.e.,
the set of rate pairs that can be reliably transmitted over the
network under the constraint that the receiver recover the state
sequence ``almost losslessly,'' i.e., with vanishing symbol error
rate. The coding scheme we propose 
is a Block-Markov scheme with backward decoding. Key is that, using
superposition coding (as in \cite{slepian1973noiseless}), the
transmitters cooperate to transmit some of the past state information
that they share.  To guarantee that, once transmission terminates, the
backward decoding can commence, we assume that the capacity region of
the SD-MAC (with causal side-information but without any
state-reconstruction constraints) has a nonempty interior.

It is noteworthy that introducing the state-reconstruction constraints
makes the problem more tractable: The capacity region of the SD-MAC
with causal state information (and no state-reconstruction
constraints) is still unknown. It was studied in
\cite{lapidoth2012common,lapidoth2012double,li2012multiple}, where it was shown that
Shannon strategies are not optimal and that strictly-causal state
information is beneficial.

The literature on state-dependent channels is vast. We
mention here a small fraction that focuses on both the data to be
transmitted and the channel state.
Sutivong et al. \cite{sutivong2005channel} studied joint message
transmission and state estimation over Gaussian channels with
noncausal encoder state information. Kim et
al. \cite{kim2008state} considered state amplification for both discrete and continuous
channels with noncausal state information. Zhang et al. \cite{zhang2011joint}
studied joint message transmission and state estimation when the
encoders know only the state distribution but not its realization;
they also extended their results to multiple-access channels.

More recently, Li et al. \cite{li2024achievable} derived an achievable
rate-distortion region for joint state and message communication over
multiple-access channels with causal state information and generalized
feedback.  Subsequent work \cite{li2025simultaneous} developed a
simultaneous-decoding approach and obtained improved achievable
regions for multiuser channels. While our setting can be viewed as a
special case of theirs (when the generalized feedback is useless, and
the allowed distortion is zero), we provide an exact characterization
of the capacity region for multiple-access channels under the
nearly-lossless state-reconstruction requirement.

A complementary line of research concerns state masking, which seeks
to limit the information about the state revealed to the receiver
\cite{merhav2007information}. Le Treust and Bloch \cite{le2020state}
studied both state masking and amplification within a unified
framework. Also related are studies of integrated sensing and
communication in which, rather than by the receivers, the state
estimation is performed by the transmitters (utilizing channel
feedback) \cite{kobayashi2019joint,ahmadipour2022information,
  liu2025fundamental}.

The rest of the paper is organized as follows. After some
notation, we turn to Section~\ref{sec:setup}, which describes the
communication scenario in greater detail and which defines the
capacity region we seek. Section~\ref{sec:main} presents the main
theorem characterizing said region. The proof is presented in
Section~\ref{sec:converse} (converse) and
Section~\ref{sec:seq-achievability} (achievability).

\subsection{Notation}
When $k$ is a positive integer, $[1:k]$ denotes the set
$\{1, \ldots, k\}$. More generally, for any positive real number
$a > 0$, we use $[1:a]$ to denote the (possibly empty) set
$\{1,2, \ldots, \lfloor a \rfloor \}$. Upper-case calligraphic fonts
denote finite sets, e.g., $\set{S}$ denotes the set of possible
channel states. The $n$-fold Cartesian product of a set $\set{A}$ with
itself is denoted~$\set{A}^{\times n}$ and comprises all $n$-tuples
with entries from $\set{A}$. The cardinality of a set $\set{A}$ is
denoted~$|\set{A}|$, so $|\set{S}|$ denotes the number of channel
states. A random variable taking values in the set $\set{A}$ is
usually denoted~$A$ and its Probability Mass Function (PMF) $P_{A}$.
Hence, the random channel state $S$ takes values in $\set{S}$ and its
PMF is $P_{S}$. A generic realization of such a random variable is
usually denoted~$a$, so $P_{S}(s)$ is the probability that $S$ equals
$s$. Given a PMF $P_X$ and a conditional PMF $P_{Y|X}$, we write
$P_X \, P_{Y|X}$ for the joint PMF that assigns the pair $(x,y)$ the
probability $P_{X}(x) \, P_{Y|X}(y|x)$.
The indicator function of an event $ \mathcal{E} $ is denoted
$ \mathbb{I}\{\mathcal{E}\} $ and equals 1 if $ \mathcal{E} $ holds
and equals 0 otherwise. For example, $ \mathbb{I}\{x = y\} $ is 1 if
$ x = y $, and is 0 otherwise.

The empirical PMF $P_{a^n}$ of a sequence $a^n \in \set{A}^{\times n}$
is defined as
\begin{IEEEeqnarray}{c}
 P_{a^n}(a)
 \triangleq
 \frac{1}{n}\sum_{i=1}^n \mathbb{I}\{a_i=a\},
 \qquad a\in\set{A}.
 \label{eq:empirical-pmf}
\end{IEEEeqnarray}
Given a PMF $P_A$ on $\set{A}$
and some $\delta>0$, the strongly $\delta$-typical set $\mathcal T_\delta^{(n)}(P_A)$ is defined as in \cite[Definition~4.1]{moser2025advanced}:
\begin{IEEEeqnarray}{rCl}
 \mathcal T_\delta^{(n)}(P_A)
 &\triangleq&
 \biggl\{a^n\in\set{A}^{\times n}:
 \Bigl( |P_{a^n}(a)-P_A(a)| < \frac{\delta}{|\set{A}|}, \quad \forall
 a \in\set{A} \Bigr) \; \text{and} \;
 \nonumber\\
 &&\hspace{39mm}
 \Bigl( (P_A(a)=0) \implies (P_{a^n}(a)=0 ) \Bigr) \biggr\}.
 \label{eq:strong-typical-set}
\end{IEEEeqnarray}
We abbreviate $\mathcal T_\delta^{(n)}(P_A)$ as
$\mathcal T_\delta^{(n)}(A)$. More generally,
$\mathcal T_\delta^{(n)}(A_1,\ldots,A_m)$ denotes the strongly typical
set with respect to the joint PMF $P_{A_1\cdots A_m}$ on
$\set{A}_1\times\cdots\times\set{A}_m$. 

All logarithms are of base two. The binary entropy function is denoted
by $H_{\mathrm b}(p)$.

\section{The Setup}\label{sec:setup}

We consider a two-to-one SD-MAC with finite input alphabets
$\set{X}_1$, $\set{X}_2$ and with a finite output alphabet
$\set{Y}$. The state alphabet $\set{S}$ is finite and has the form
\begin{IEEEeqnarray}{rCl}
\set{S} & = & \set{S}_0 \times \set{S}_1 \times \set{S}_2  
\end{IEEEeqnarray}
so the time $i$ state $S_{i}$ is a triple
\begin{IEEEeqnarray}{rCl}
  S_{i} & = & \bigl( S_{0,i},S_{1,i},S_{2,i} \bigr).
\end{IEEEeqnarray}
The states $\{S_{i}\}$ are IID $\sim P_{S}$, with $P_{S}$ having the product
form
\begin{IEEEeqnarray}{rCl}
  P_{S} & = & P_{S_0} \, P_{S_1} \, P_{S_2}.
\end{IEEEeqnarray}
The three components $S_{0,i}$, $S_{1,i}$, $S_{2,i}$ of the Time-$i$
state $S_{i}$ are thus independent and of PMFs $P_{S_0}$, $P_{S_1}$,
and $P_{S_2}$. By possibly redefining the state alphabet, we may
assume---as we shall---that $P_{S}$ is of full support (i.e., that
$P_{S}(s)$ is positive for every state $s \in \set{S}$).
Conditional on the states, the channel is memoryless, and we denote
its law by $P_{Y|X_1X_2S_0S_1S_2}$ or $W_{Y|X_1X_2S_0S_1S_2}$.

Each encoder wishes to convey a message to the decoder: Encoder~1
wishes to convey $M_{1}$ and Encoder~2 $M_{2}$. The two messages are
independent, with $M_{1}$ uniformly distributed over the set
$\set{M}_{1} = \bigl[1:2^{nR_{1}}\bigr]$, and with $M_{2}$ uniformly
distributed over $\set{M}_{2} = \bigl[1:2^{nR_{2}}\bigr]$. Here $n$
denotes the blocklength, and $R_{1}$, $R_{2}$ denote the rates of
transmission in bits per channel-use.

The sequences $\{S_{0,i}\}$ and $\{S_{1,i}\}$ are revealed to
Encoder~1 causally, so the symbol
$X_{1,i} = X_{1,i}(M_{1},S_{0}^{i},S_{1}^{i})$ it produces at Time-$i$
can depend not only on its message $M_{1}$ but also on the
present-and-past values of the sequences $\{S_{0,i}\}$ and
$\{S_{1,i}\}$. Encoder~1 is thus specified by $n$ functions
\begin{subequations}
  \label{block:def_f_no_cribbing}
\begin{IEEEeqnarray}{rCl}
    f_{1,i}:\set{M}_1\times \set{S}_0^{\times i}
    \times\set{S}_1^{\times i}\to
    \set{X}_1, \qquad i \in [1:n]
  \end{IEEEeqnarray}
  where $n$ is the blocklength.  Likewise, Encoder~2, which is
  revealed the sequences $\{S_{0,i}\}$ and $\{S_{2,i}\}$ causally,
  produces the Time-$i$ symbol
  $X_{2,i} = X_{2,i}(M_{2},S_{0}^{i},S_{2}^{i})$ in accordance with
  the $n$~functions
\begin{IEEEeqnarray}{rCl}
    f_{2,i}:\set{M}_2\times \set{S}_0^{\times i}
    \times\set{S}_2^{\times i}\to
    \set{X}_2, \qquad  i \in [1:n].
\end{IEEEeqnarray}
\end{subequations}

Based on the received output sequence $Y^n$, the receiver wishes to
recover the messages as well as the state sequence. It thus has the form
\begin{IEEEeqnarray}{c}
  \label{eq:def_decoding}
  \phi\colon \set{Y}^{\times n}\to
  \set{M}_1\times\set{M}_2\times \set{S}^{\times n}
\end{IEEEeqnarray}
with the understanding that it maps $Y^{n}$ to
$\bigl(\hat{M}_{1},\hat{M}_{2}, \hat{S}_{0}^{n}, \hat{S}_{1}^{n},
\hat{S}_{2}^{n}\bigr)$.

The message-error decoding probability is 
\begin{IEEEeqnarray}{c}
    P_\textnormal{e}^{(n)} \triangleq \text{Pr}\Bigl((\hat{M}_1,\hat{M}_2) \neq (M_1, M_2) \Bigr)
\end{IEEEeqnarray}
and the state-reconstruction symbol error-rate is
\begin{IEEEeqnarray}{c}
    \beta^{(n)} \triangleq \frac{1}{n}\sum_{i=1}^n
    \Pr\Bigl((S_{0,i},S_{1,i},S_{2,i})
    \neq(\hat S_{0,i},\hat S_{1,i},\hat S_{2,i})\Bigr).
\end{IEEEeqnarray}

A rate pair $(R_1,R_2)\in\mathbb{R}_{\geq 0}^2$ is said to be
\emph{achievable} if, for some $\epsilon_0 > 0$, there exists a sequence
of rate-$(R_{1}+\epsilon_0,R_{2} + \epsilon_0)$ encoders and decoders
(indexed by the blocklength $n$) such that both $P^{(n)}_\text{e}$ and
$\beta^{(n)}$ tend to zero. The \emph{capacity region} $\set{C}$ is
the closure of the set of all achievable rate pairs.

The added $\epsilon_0$ in the definition allows us to state and prove
the main theorem without the 
intricacies of the zero-rate case.

\begin{remark}
  \label{rem:convex}
  The capacity region is a (possibly empty) convex compact subset of
  $\mathbb{R}^{2}$.
\end{remark}
\begin{proof}
  The set is closed by definition; it is bounded because the alphabets are
  finite; and convexity can be established using a time-sharing argument.
  \end{proof}

\section{Main Result}\label{sec:main}
\newcommand{\CI}{{\mathcal C}^{(I)}}

The main result of this paper is a characterization of the capacity
region. It is derived, under the following assumption:
\begin{assumption}\label{ass:nontri}
  The SD-MAC $W(y|x_{1},x_{2},s_{0},s_{1},s_{2})$ is such that, in the
  absence of a state-reconstruction requirement (i.e., when
  $\beta^{(n)}$ need not tend to zero), there exists a
  positive rate pair $(R_1,R_2)\in\mathbb R_{>0}^2$ that is achievable
  (when the encoders are causally informed).
\end{assumption}

\begin{remark}
  \label{rem:common_state}
  In the ``common state'' case---i.e., when $S_{1}$ and $S_{2}$ are
  deterministic---Assumption~\ref{ass:nontri} is
  equivalent to the assumption that some positive rate pair can be
  achieved over the SD-MAC using Shannon strategies. This equivalence need
  not hold in general.
\end{remark}
\begin{IEEEproof} See
  Appendix~\ref{app:proof_common_state}. \end{IEEEproof}

To describe the capacity region, we need some preliminary definitions.
Let $\mathscr P$ denote the collection of all joint PMFs of the
following form, where $Q$ and $V$ are finite auxiliary random variables:
\begin{IEEEeqnarray}{c}\label{eq:pmf}
    P_Q \, P_{S_0} \, P_{S_1} \, P_{S_2} \, P_{V|Q} \,
    P_{X_1|S_0S_1VQ} \, P_{X_2|S_0S_2VQ} \,
        W_{Y|X_1X_2S_0S_1S_2}.
\end{IEEEeqnarray}
For each $P\in\mathscr P$, let $\set R_{<}(P)$ denote the rate polytope
\begin{IEEEeqnarray}{c}\label{eq:region} \left\{
    (R_1,R_2)\in\mathbb R_{\geq0}^2
    \middle|
    \begin{aligned}
    R_1<&I(X_1,S_1;Y\,|\,X_2,S_0,S_2,V,Q)-H(S_1),\\
    R_2<&I(X_2,S_2;Y\,|\,X_1,S_0,S_1,V,Q)-H(S_2),\\
    R_1+R_2<&I(X_1,X_2,S_1,S_2;Y\,|\,S_0,V,Q)
        -H(S_1)-H(S_2),\\
    R_1+R_2<&I(X_1,X_2,S_0,S_1,S_2;Y\,|\,Q)
        -H(S_0)-H(S_1)-H(S_2)
    \end{aligned}
    \right\}.
\end{IEEEeqnarray}
Finally, define
\begin{IEEEeqnarray}{c}
\CI
\triangleq
\operatorname{cl}\left(
\bigcup_{P\in\mathscr P}\set R_{<}(P)
\right)
\label{eq:single-letter-region}
\end{IEEEeqnarray}
where $\operatorname{cl}(\cdot)$ stands for the closure operator.

\begin{remark}[Cardinality bounds]\label{rem:cardinality}
Without altering $\CI$, we may restrict the cardinalities of the auxiliary alphabets to
\begin{IEEEeqnarray}{rCl}
    |\set{Q}|&\leq&5,\\
    |\set{V}|&\leq&|\set{Y}|+4.
\end{IEEEeqnarray}
\end{remark}
\begin{IEEEproof} See Appendix~\ref{app:proof-cardinality}. \end{IEEEproof}

\begin{theorem}\label{thm:main}
  Under Assumption~\ref{ass:nontri},
  \begin{IEEEeqnarray}{rCl}
    \set{C} = \CI.
  \end{IEEEeqnarray}
\end{theorem}

\section{Converse}\label{sec:converse}

We first prove the converse. Fix an achievable rate pair $(R_1,R_2)$.
By definition, there exist a constant $\epsilon_0>0$ and a sequence of
rate-$(R_1+\epsilon_0,R_2+\epsilon_0)$ coding schemes such that
$P_{\textnormal e}^{(n)}\to0$ and $\beta^{(n)}\to0$. Define
\begin{IEEEeqnarray}{rCl}
    \delta_{\text e}^{(n)}
    &\triangleq&
    (R_1+R_2+2\epsilon_0)P_\textnormal{e}^{(n)}
    +\frac{1}{n}H_{\text b}\bigl(P_\textnormal{e}^{(n)}\bigr)\\
    \delta_{\beta}^{(n)}
    &\triangleq&
    \beta^{(n)}\log(|\set{S}_0||\set{S}_1||\set{S}_2|-1)
    +H_{\text b}\bigl(\beta^{(n)}\bigr)\\
    \eta_n&\triangleq&\delta_{\text e}^{(n)}+\delta_{\beta}^{(n)}.
\end{IEEEeqnarray}
Since $P_{\textnormal e}^{(n)}$ and $\beta^{(n)}$ tend to zero, so does $\eta_n$.
Applying Fano's inequality \cite[Proposition~2.17]{moser2025advanced}
to the messages, we obtain
\begin{IEEEeqnarray}{rCl}
    H(M_1,M_2\,|\,Y^n)
    &\leq&
    n\,\delta_{\text e}^{(n)}.
    \label{eq:three-fano-messages}
\end{IEEEeqnarray}
Define
\begin{IEEEeqnarray}{rCl}
p_i^{(n)}
&\triangleq&
\Pr\bigl( (S_{0,i},S_{1,i},S_{2,i})\neq(\hat{S}_{0,i},\hat{S}_{1,i},\hat{S}_{2,i})\bigr).
\end{IEEEeqnarray}
Since $(\hat{S}_{0,i},\hat{S}_{1,i},\hat{S}_{2,i})$ are functions of $Y^n$,
\begin{IEEEeqnarray}{rCl}
\IEEEyesnumber
\IEEEyessubnumber*
H(S_{0,i},S_{1,i},S_{2,i}\,|\, Y^n)
&\leq&
H(S_{0,i},S_{1,i},S_{2,i}\,|\,\hat{S}_{0,i},\hat{S}_{1,i},\hat{S}_{2,i})
\\
&\leq&
H_{\text b}\bigl(p_i^{(n)}\bigr)
+p_i^{(n)}
\log\bigl(
|\set{S}_0||\set{S}_1||\set{S}_2|-1
\bigr),
\label{eq:coordinatewise-fano}
\end{IEEEeqnarray}
where the second inequality follows from Fano's inequality.

By the chain rule, the fact that conditioning reduces entropy, and the
concavity of $H_{\text b}(\cdot)$,
\begin{IEEEeqnarray}{rCl}\label{eq:three-fano-states}
\IEEEyesnumber
\IEEEyessubnumber*
    &&H(S_0^n,S_1^n,S_2^n\,|\,Y^n)
    \nonumber\\
    &=&
    \sum_{i=1}^n
    H(S_{0,i},S_{1,i},S_{2,i}\,|\,
      S_0^{i-1},S_1^{i-1},S_2^{i-1},Y^n)
    \\
    &\leq&
    \sum_{i=1}^n
    H(S_{0,i},S_{1,i},S_{2,i}\,|\,Y^n)
    \\
    &\leq&
    \sum_{i=1}^n H_{\text b}\bigl(p_i^{(n)}\bigr)
    +\sum_{i=1}^n p_i^{(n)}\log\bigl(|\set{S}_0||\set{S}_1||\set{S}_2|-1\bigr)
    \\
    &\leq&
    nH_{\text b}\bigl(\beta^{(n)}\bigr)
    +n\beta^{(n)}  \log\bigl(|\set{S}_0||\set{S}_1||\set{S}_2|-1\bigr)
    \\
    &=&
    n\delta_{\beta}^{(n)}.
\end{IEEEeqnarray}

From \eqref{eq:three-fano-messages} and
\eqref{eq:three-fano-states}, we obtain
\begin{IEEEeqnarray}{rCl}
    \IEEEyesnumber
\IEEEyessubnumber*&&H(M_1,M_2,S_0^n,S_1^n,S_2^n\,|\,Y^n)
    \nonumber\\
    &\leq&
    H(M_1,M_2\,|\,Y^n)
    +H(S_0^n,S_1^n,S_2^n\,|\,Y^n)
    \\
    &\leq&
    n\bigl(\delta_{\text e}^{(n)}+\delta_{\beta}^{(n)}\bigr)
    =n\eta_n.
    \label{eq:three-fano-joint}
\end{IEEEeqnarray}

Consequently, the conditional entropies obtained by revealing
$M_1$, $M_2$, or any subset of the state sequences given $Y^n$ are also at most
$n\,\eta_n$.

For each $i\in[1:n]$, define the common-state history
\begin{IEEEeqnarray}{c}\label{eq:three-history}
    V_i\triangleq S_0^{i-1}.
\end{IEEEeqnarray}
Since the state sequence is IID, $V_i$ is independent of
$(S_{0,i},S_{1,i},S_{2,i})$. Conditioned on
$(S_{0,i},S_{1,i},S_{2,i},V_i)$, the inputs $X_{1,i}$ and $X_{2,i}$ are independent, because the
private state histories and messages are independent. Hence, the factorization
\begin{IEEEeqnarray}{c}\label{eq:three-factorization-n}
    P_{X_{1,i}X_{2,i}|S_{0,i}S_{1,i}S_{2,i}V_i}
    =P_{X_{1,i}|S_{0,i}S_{1,i}V_i} \,
     P_{X_{2,i}|S_{0,i}S_{2,i}V_i}.
\end{IEEEeqnarray}

Since the message $M_{1}$ that Encoder~1 wishes to convey and its
private state $S_{1}^{n}$ are independent and as a pair are independent of $(M_2,S_0^n,S_2^n)$
\begin{IEEEeqnarray}{rCl}
\IEEEyesnumber
\IEEEyessubnumber*
    &&n\bigl(R_1+\epsilon_0+H(S_1)\bigr) \nonumber\\
    =&&H(M_1,S_1^n\,|\,M_2,S_0^n,S_2^n)\label{eq:three-r1a}\\
    \leq&&I(M_1,S_1^n;Y^n\,|\,M_2,S_0^n,S_2^n)+n\,\eta_n
        \label{eq:three-r1b}\\
    =&&\sum_{i=1}^n\Bigl[
        H(Y_i\,|\,Y^{i-1},M_2,S_0^n,S_2^n)
        -H(Y_i\,|\,Y^{i-1},M_1,M_2,S_0^n,S_1^n,S_2^n)
        \Bigr]+n\,\eta_n\label{eq:three-r1c}\\
    \leq&&\sum_{i=1}^n\Bigl[
        H(Y_i\,|\,X_{2,i},S_{0,i},S_{2,i},V_i)
        -H(Y_i\,|\,X_{1,i},X_{2,i},S_{0,i},S_{1,i},S_{2,i},V_i)
        \Bigr]+n\,\eta_n\label{eq:three-r1d}\\
    =&&\sum_{i=1}^n
        I(X_{1,i},S_{1,i};Y_i\,|\,
        X_{2,i},S_{0,i},S_{2,i},V_i)+n\,\eta_n,
        \label{eq:three-r1e}
\end{IEEEeqnarray}
where~\eqref{eq:three-r1d} holds because $X_{2,i}$ is a function of
$(M_2,S_0^i,S_2^i)$ and because the channel is memoryless.

By symmetry,
\begin{IEEEeqnarray}{rCl}\label{eq:three-r2}
    n\bigl(R_2+\epsilon_0+H(S_2)\bigr)&\leq&
    \sum_{i=1}^n I(X_{2,i},S_{2,i};Y_i\,|\,
    X_{1,i},S_{0,i},S_{1,i},V_i)+n\,\eta_n.
\end{IEEEeqnarray}

Conditioning on the common state sequence yields the private-state sum bound
\begin{IEEEeqnarray}{rCl}
\IEEEyesnumber
\IEEEyessubnumber*
    &&n\bigl(R_1+R_2+2\epsilon_0+H(S_1)+H(S_2)\bigr)\nonumber\\
    =&&H(M_1,M_2,S_1^n,S_2^n\,|\,S_0^n)
        \label{eq:three-private-suma}\\
    \leq&&I(M_1,M_2,S_1^n,S_2^n;Y^n\,|\,S_0^n)+n\,\eta_n
        \label{eq:three-private-sumb}\\
    =&&\sum_{i=1}^n H(Y_i\,|\,Y^{i-1},
    S_0^n)-H(Y_i\,|\, Y^{i-1},S_0^n,M_1,M_2,S_1^n,S_2^n)+n\,\eta_n
    \label{eq:three-private-sumc}\\
    \leq&&\sum_{i=1}^n
    I(X_{1,i},X_{2,i},S_{1,i},S_{2,i};Y_i\,|\,S_{0,i},V_i)
    +n\,\eta_n.\label{eq:three-private-sumd}
\end{IEEEeqnarray}

Finally, accounting for all three states:
\begin{IEEEeqnarray}{rCl}
\IEEEyesnumber
\IEEEyessubnumber*
    &&n\bigl(R_1+R_2+2\epsilon_0+H(S_0)+H(S_1)+H(S_2)\bigr)\nonumber\\
    =&&H(M_1,M_2,S_0^n,S_1^n,S_2^n)\label{eq:three-total-suma}\\
    \leq&&I(M_1,M_2,S_0^n,S_1^n,S_2^n;Y^n)+n\,\eta_n
        \label{eq:three-total-sumb}\\
    =&&H(Y^n)-H(Y^n\,|\,M_1,M_2,S_0^n,S_1^n,S_2^n)+n\,\eta_n
        \label{eq:three-total-sumc}\\
    \leq&&\sum_{i=1}^n\Bigl[
        H(Y_i)-H(Y_i\,|\,X_{1,i},X_{2,i},S_{0,i},S_{1,i},S_{2,i})
        \Bigr]+n\,\eta_n\label{eq:three-total-sumd}\\
    =&&\sum_{i=1}^n
        I(X_{1,i},X_{2,i},S_{0,i},S_{1,i},S_{2,i};Y_i)+n\,\eta_n.
        \label{eq:three-total-sume}
\end{IEEEeqnarray}

To single-letterize, let $J$ be uniform on $[1:n]$ and independent of
all the other random
variables, and define
\begin{IEEEeqnarray}{c}
    Q\triangleq J,\qquad V\triangleq V_J,\qquad
    (S_0,S_1,S_2,X_1,X_2,Y)
    \triangleq(S_{0,J},S_{1,J},S_{2,J},X_{1,J},X_{2,J},Y_J).
\end{IEEEeqnarray}
For each $n$, denote the induced joint PMF by $P^{(n)}$. It factorizes as
\begin{IEEEeqnarray}{c}\label{eq:three-converse-pmf}
    P_Q \, P_{S_0} \, P_{S_1} \, P_{S_2} \, P_{V|Q} \,
    P_{X_1|S_0S_1VQ} \, P_{X_2|S_0S_2VQ} \,
    W_{Y|X_1X_2S_0S_1S_2}.
\end{IEEEeqnarray}
Thus, $P^{(n)}\in\mathscr P$.

Dividing \eqref{eq:three-r1e},
\eqref{eq:three-r2}, \eqref{eq:three-private-sumd}, and
\eqref{eq:three-total-sume} by $n$ and subtracting the
entropy terms from both sides yields
\begin{IEEEeqnarray}{rCl}
\IEEEyesnumber
\IEEEyessubnumber*
    R_1&\leq&I(X_1,S_1;Y\,|\,X_2,S_0,S_2,V,Q)-H(S_1)-\epsilon_0+\eta_n\\
    R_2&\leq&I(X_2,S_2;Y\,|\,X_1,S_0,S_1,V,Q)-H(S_2)-\epsilon_0+\eta_n\\
    R_1+R_2&\leq&I(X_1,X_2,S_1,S_2;Y\,|\,S_0,V,Q)
        -H(S_1)-H(S_2)-2\epsilon_0+\eta_n\\
    R_1+R_2&\leq&I(X_1,X_2,S_0,S_1,S_2;Y\,|\,Q)
        -H(S_0)-H(S_1)-H(S_2)-2\epsilon_0+\eta_n,
\end{IEEEeqnarray}
where the quantities on the right-hand sides are evaluated under
$P^{(n)}$.

Since $\eta_n\to0$, we can choose $n$ sufficiently large so that
$\eta_n<\epsilon_0$. All four strict inequalities in \eqref{eq:region} then hold under
$P^{(n)}$, and hence
\begin{IEEEeqnarray}{c}
  \label{eq:amos_contained}
(R_1,R_2)\in\set R_{<}\bigl(P^{(n)}\bigr)
\subseteq
\bigcup_{P\in\mathscr P}\set R_{<}(P).
\end{IEEEeqnarray}
Thus, every achievable rate pair belongs to the above union, and
$\set{C}$---the closure of the set of all achievable rates---is thus
contained in the closure of the RHS of \eqref{eq:amos_contained},
i.e., in $\CI$.

\section{Achievability}
\label{sec:seq-achievability}
Given some fixed joint PMF of the form
\begin{IEEEeqnarray}{c}
 P_Q \, P_{S_0} \, P_{S_1} \, P_{S_2} \, P_{V|Q} \,
 P_{X_1|S_0S_1VQ} \, P_{X_2|S_0S_2VQ} \,
 P_{Y|X_1X_2S_0S_1S_2}
 \label{eq:seq-pmf}
\end{IEEEeqnarray}
we need to establish the achievability of the polytope $\set R_{<}(P)$,
which \emph{we henceforth assume is nonempty}. To
that end, we first express it as the
marginal distribution of a distribution involving Shannon strategies.
\subsection{Introducing Shannon Strategies}
By the Functional Representation Lemma
\cite[Appendix~B]{el2011network}, there exist finite-alphabet
random variables $T_1,T_2$ of some conditional PMFs $P_{T_1|VQ}$ and
$P_{T_2|VQ}$, and some deterministic functions $x_1(\cdot)$,
$x_2(\cdot)$ such that
\begin{IEEEeqnarray}{c}
 X_1 = x_1(T_1,S_0,S_1),\qquad
 X_2 = x_2(T_2,S_0,S_2).
 \label{eq:seq-strategies}
\end{IEEEeqnarray}
More specifically, the joint PMF of $(T_{1},T_{2})$ and the
other variables has the form 
\begin{IEEEeqnarray}{rCl}
P_{QS_0S_1S_2VT_1T_2X_1X_2Y}
&=&
P_Q\,P_{S_0}\,P_{S_1}\,P_{S_2}\,P_{V|Q}\,
P_{T_1|VQ}\,P_{T_2|VQ}
\nonumber\\
\cdot
\mathbb{I}\bigl\{X_1=&x_1&(T_1,S_0,S_1)\bigr\} \,
\mathbb{I}\bigl\{X_2=x_2(T_2,S_0,S_2)\bigr\} \,
P_{Y|X_1X_2S_0S_1S_2}.
\label{eq:seq-strategy-pmf}
\end{IEEEeqnarray}
In particular, 
the following are Markov chains:
\begin{IEEEeqnarray}{rCl}
\IEEEyesnumber
\IEEEyessubnumber*
(S_0,S_1,S_2)
&\markov&
(V,Q)\markov(T_1,T_2)\label{eq:strategy_markova}\\
T_1
&\markov&
(V,Q)\markov T_2\label{eq:strategy_markovb}\\
(V,Q,T_1,T_2)
&\markov&
(X_1,X_2,S_0,S_1,S_2)\markov Y\label{eq:strategy_markovc}\\
T_1
&\markov&
(V,Q,X_1,S_0,S_1)\markov Y\label{eq:strategy_markovd}\\
T_2
&\markov&
(V,Q,X_2,S_0,S_2)\markov Y\label{eq:strategy_markove}\\
(T_1,S_1)&\markov& (V,Q,S_0) \markov (T_2,S_2).\label{eq:strategy_markovf}
\end{IEEEeqnarray}

As we next argue, because $X_k$ is a deterministic function of
$(T_k,S_0,S_k)$, for each $k \in \{1,2\}$, and because the channel output
depends on $T_1,T_2$ only through $X_1,X_2$, 
\begin{IEEEeqnarray}{rCl} \label{eq:seq-strategy-identities}
\IEEEyesnumber
\IEEEyessubnumber*
 I(T_1,S_1;Y\,|\,T_2,S_0,S_2,V,Q)
 &=&I(X_1,S_1;Y\,|\,X_2,S_0,S_2,V,Q),\\
 I(T_2,S_2;Y\,|\,T_1,S_0,S_1,V,Q)
 &=&I(X_2,S_2;Y\,|\,X_1,S_0,S_1,V,Q),\\
 I(T_1,T_2,S_1,S_2;Y\,|\,S_0,V,Q)
 &=&I(X_1,X_2,S_1,S_2;Y\,|\,S_0,V,Q),\\
 I(V,T_1,T_2,S_0,S_1,S_2;Y\,|\,Q)
 &=&I(X_1,X_2,S_0,S_1,S_2;Y\,|\,Q).
\end{IEEEeqnarray}

Indeed, the first identity in \eqref{eq:seq-strategy-identities} holds because
\begin{IEEEeqnarray}{rCl}
\IEEEyesnumber
\IEEEyessubnumber*
&&I(T_1,S_1;Y\,|\,T_2,S_0,S_2,V,Q)\nonumber\\
&\stackrel{(a)}{=}&
H(Y\,|\,T_2,X_2,S_0,S_2,V,Q)
-H(Y\,|\,T_1,X_1,T_2,X_2,S_0,S_1,S_2,V,Q)\\
&\stackrel{(b)}{=}&
H(Y\,|\,X_2,S_0,S_2,V,Q)
-H(Y\,|\,X_1,X_2,S_0,S_1,S_2,V,Q)\\
&=&
I(X_1,S_1;Y\,|\,X_2,S_0,S_2,V,Q).
\label{eq:seq-strategy-identity-first-proof}
\end{IEEEeqnarray}
Here, $(a)$ follows from the fact that $X_1$ and $X_2$ are deterministic
functions of $(T_1,S_0,S_1)$ and $(T_2,S_0,S_2)$, respectively, and $(b)$ follows from the Markov chains \eqref{eq:strategy_markovc} and \eqref{eq:strategy_markove}. 
The remaining identities in
\eqref{eq:seq-strategy-identities} follow analogously.

We conclude that $\set R_{<}(P)$ can be expressed as
\begin{IEEEeqnarray}{c}\label{eq:Tpolytope} \left\{
    (R_1,R_2)\in\mathbb R_{\geq0}^2
    \middle|
    \begin{aligned}
    R_1<&I(T_1,S_1;Y\,|\,T_2,S_0,S_2,V,Q)-H(S_1),\\
    R_2<&I(T_2,S_2;Y\,|\,T_1,S_0,S_1,V,Q)-H(S_2),\\
    R_1+R_2<&I(T_1,T_2,S_1,S_2;Y\,|\,S_0,V,Q)
        -H(S_1)-H(S_2),\\
    R_1+R_2<&I(T_1,T_2,S_0,S_1,S_2;Y\,|\,Q)
        -H(S_0)-H(S_1)-H(S_2)
    \end{aligned}
    \right\}.
\end{IEEEeqnarray}

\subsection{The Polytope and its Vertices}
\label{sec:seq-net-rates}

For a fixed PMF $P$ as in \eqref{eq:seq-strategy-pmf}, define
\begin{IEEEeqnarray}{rCl} \label{eq:seq-region-parameters}
 \IEEEyesnumber
 \IEEEyessubnumber*
A_1&\triangleq&I(T_1,S_1;Y\,|\,V,T_2,S_0,S_2,Q)-H(S_1),\\
 A_2&\triangleq&I(T_2,S_2;Y\,|\,V,T_1,S_0,S_1,Q)-H(S_2),\\
 C_{\rm p}
 &\triangleq&I(T_1,T_2,S_1,S_2;Y\,|\,V,S_0,Q)
              -H(S_1)-H(S_2),\\
 C_{\rm t}
 &\triangleq&I(V,T_1,T_2,S_0,S_1,S_2;Y\,|\,Q)
              -H(S_0)-H(S_1)-H(S_2),\\
 C&\triangleq&\min\{C_{\rm p},C_{\rm t}\}.
\end{IEEEeqnarray}
With these definitions, we can express the nonempty polytope $\set R_<(P)$ of~\eqref{eq:Tpolytope}
as comprising the rate pairs satisfying
\begin{IEEEeqnarray}{c}
 0 \leq R_1< A_1,\qquad 0 \leq R_2< A_2,\qquad R_1+R_2< C.
 \label{eq:seq-fixed-region}
\end{IEEEeqnarray}
We note that
\begin{IEEEeqnarray}{rCl}
  \label{eq:printer_works}
  C & \leq & A_1+A_2.
\end{IEEEeqnarray}
Indeed, 
the stronger inequality $C_{\rm p} \leq A_1+A_2$ holds because
\begin{IEEEeqnarray}{rCl}
 A_1+A_2-C_{\rm p}
 &=&I(T_1,S_1;T_2,S_2\,|\,Y,V,S_0,Q), 
 \label{eq:seq-submodularity}
\end{IEEEeqnarray}
as can be verified using the chain rule and the Markov condition~\eqref{eq:strategy_markovf}.

Let
$\mathscr P_+\triangleq
\{P\in\mathscr P\,:\, \set R_{<}(P)\neq\emptyset\}$
denote the subfamily of PMFs that induce nonempty strict rate
polytopes. From~\eqref{eq:printer_works}, we conclude that, for every
$P\in\mathscr P_+$, the polytope $\set R_{<}(P)$ has one of the four
shapes depicted in Fig.~\ref{fig:seq-ach-reg}. Its closure is the
polytope $\set R_{\leq}(P)$ obtained by replacing the strict
inequalities in~\eqref{eq:seq-fixed-region} with weak inequalities.

The PMFs in $\mathscr P\setminus\mathscr P_+$ contribute only empty
sets to~\eqref{eq:single-letter-region}. Moreover, replacing each set
in a family by its closure does not change the closure of their union.
Consequently,
\begin{IEEEeqnarray}{rCl}
\CI
&=&
\operatorname{cl}\left(
\bigcup_{P\in\mathscr P_+}\set R_{\leq}(P)
\right).
\label{eq:NEW_single-letter-region}
\end{IEEEeqnarray}

\begin{figure}
\begin{tikzpicture}[
    scale=3,
    axis/.style={-latex, thick},
    bound/.style={black, thick},
    fillarea/.style={fill=black!20, opacity=0.6},
    dashline/.style={dashed, gray, thin},
    lbl/.style={font=\small},
    dominant/.style={black, ultra thick},
]
\matrix (m) [matrix of nodes, column sep=1cm, row sep=1cm] {
    \begin{tikzpicture}
        \draw[axis] (-0.5,0) -- (3,0) node[right] {$R_1$};
        \draw[axis] (0,-0.5) -- (0,3) node[above] {$R_2$};
        \coordinate (O) at (0,0);
        \coordinate (C1) at (2,0.5);
        \coordinate (C2) at (0.7,1.8);
        \fill[fillarea] (O) -- (2,0) -- (C1) -- (C2) -- (0,1.8) -- cycle;
        \draw[bound] (2,0) -- (C1);
        \draw[bound] (0,1.8) -- (C2);
        \draw[dominant] (C1) -- (C2);
        \draw[dashline] (2,0.5) -- (2,2.5) node[above, lbl] {$A_1$};
        \draw[dashline] (0.7,1.8) -- (2.5,1.8) node[right, lbl] {$A_2$};
        \draw[dashline] (2.5,0) node[below, lbl] {$C$} -- (0,2.5) node[left, lbl] {$C$};
        \draw[draw=none, fill=blue, opacity=0.2] (C1) circle (0.2);
        \draw[draw=none, fill=black] (C2) circle (0.1);
        \node[below left] at (O) {$0$};
        \node[above, font=\small\bfseries] at (1.2,-1.5)
            {Subfigure 2.1: $C\geq A_1$, $C\geq A_2$};
    \end{tikzpicture}
    &
    \begin{tikzpicture}
        \draw[axis] (-0.5,0) -- (3,0) node[right] {$R_1$};
        \draw[axis] (0,-0.5) -- (0,3) node[above] {$R_2$};
        \coordinate (O) at (0,0);
        \coordinate (A) at (1.2,0);
        \coordinate (C1) at (1.2,0.6);
        \coordinate (C2) at (0,1.8);
        \fill[fillarea] (O) -- (A) -- (C1) -- (C2) -- cycle;
        \draw[bound] (A) -- (C1);
        \draw[dominant] (C1) -- (C2);
        \draw[dashline] (1.2,0.6) -- (1.2,2.5) node[above, lbl] {$A_1$};
        \draw[dashline] (-1,2.2) node[left, lbl] {$A_2$} -- (2,2.2);
        \draw[dashline] (1.8,0) node[below, lbl] {$C$} -- (-0.8,2.6) node[left, lbl] {$C$};
        \draw[draw=none, fill=blue, opacity=0.2] (C1) circle (0.2);
        \draw[draw=none, fill=black] (C2) circle (0.1);
        \node[below left] at (O) {$0$};
        \node[above, font=\small\bfseries] at (1.2,-1.5)
            {Subfigure 2.2: $C\geq A_1$, $C<A_2$};
    \end{tikzpicture}
    \\
    \begin{tikzpicture}
        \draw[axis] (-0.5,0) -- (3,0) node[right] {$R_1$};
        \draw[axis] (0,-0.5) -- (0,3) node[above] {$R_2$};
        \coordinate (O) at (0,0);
        \coordinate (A) at (1.8,0);
        \coordinate (C1) at (0.6,1.2);
        \coordinate (C2) at (0,1.2);
        \fill[fillarea] (O) -- (A) -- (C1) -- (C2) -- cycle;
        \draw[bound] (C2) -- (C1);
        \draw[dominant] (C1) -- (A);
        \draw[dashline] (2.2,-0.8) -- (2.2,1.5) node[above, lbl] {$A_1$};
        \draw[dashline] (0.6,1.2) -- (2.7,1.2) node[right, lbl] {$A_2$};
        \draw[dashline] (2.3,-0.5) node[right, lbl] {$C$} -- (0,1.8) node[left, lbl] {$C$};
        \draw[draw=none, fill=black] (C1) circle (0.1);
        \draw[draw=none, fill=red, opacity=0.2] (A) circle (0.2);
        \node[below left] at (O) {$0$};
        \node[above, font=\small\bfseries] at (1.2,-1.5)
            {Subfigure 2.3: $C<A_1$, $C\geq A_2$};
    \end{tikzpicture}
    &
    \begin{tikzpicture}
        \draw[axis] (-0.5,0) -- (3,0) node[right] {$R_1$};
        \draw[axis] (0,-0.5) -- (0,3) node[above] {$R_2$};
        \coordinate (O) at (0,0);
        \coordinate (A) at (1.5,0);
        \coordinate (B) at (0,1.5);
        \fill[fillarea] (O) -- (A) -- (B) -- cycle;
        \draw[dominant] (A) -- (B);
        \draw[dashline] (2.2,-0.8) node[right, lbl] {$A_1$} -- (2.2,2.5);
        \draw[dashline] (-0.8,2.2) node[left, lbl] {$A_2$} -- (2.5,2.2);
        \draw[dashline] (2.3,-0.8) -- (-0.8,2.3) node[above, lbl] {$C$};
        \draw[draw=none, fill=red, opacity=0.2] (A) circle (0.2);
        \draw[draw=none, fill=black] (B) circle (0.1);
        \node[below left] at (O) {$0$};
        \node[above, font=\small\bfseries] at (1.2,-1.5)
            {Subfigure 2.4: $C<A_1$, $C<A_2$};
    \end{tikzpicture}
    \\
};
\end{tikzpicture}
\centering
\captionsetup{justification=centering}
\caption{The four possible shapes of a nonempty polytope $\set
  R_{\leq}(P)$. By symmetry, to establish that the endpoints of the dominant face
  are all in $\set{C}$, it suffices to consider the colored vertices.
  The coding scheme for the blue vertices is described in
  Section~\ref{sec:blue_achie}, and for the red ones in
  Section~\ref{sec:red_achie}.}
\label{fig:seq-ach-reg}
\end{figure}

Let $\mathbf v^{(1)}$ and $\mathbf v^{(2)}$ denote the two endpoints
of the dominant face of $\set R_{\leq}(P)$. We will show that both
are in~$\set{C}$ by establishing that every neighborhood of each of
them contains an achievable rate pair. The convexity of $\set{C}$ will
then establish that the dominant face is in~$\set{C}$, from which it
will follow that the entire polytope is in~ $\set{C}$.

By symmetry, it suffices to consider the following two types of endpoints:
\begin{enumerate}
 \item the blue vertices $(A_1,C-A_1)$ in Subfigures \ref{fig:seq-ach-reg}.1 and \ref{fig:seq-ach-reg}.2  when $C> A_1$;
 \item the red vertices $(C,0)$ in Subfigures \ref{fig:seq-ach-reg}.3 and \ref{fig:seq-ach-reg}.4  when $C\leq A_1$.
\end{enumerate}
The coding scheme for the blue vertices is described in
  Section~\ref{sec:blue_achie}, and for the red ones in
  Section~\ref{sec:red_achie}.
  
\subsection{An Overview of the Block-Markov Scheme}

We next turn to our proposed Block-Markov scheme with backward
decoding. In each length-$l$ block we draw random bins for the
different (typical) state sequences and use superposition coding. We
draw $2^{lR_0^{\mathrm b}}$ ``$S_{0}$-bins'' that will be used in the
next block to describe the $S_{0}$ state sequence prevailing during
the current block; $2^{lR_1^{\mathrm b}}$ ``$S_{1}$-bins'' that
will be used to describe the $S_{1}$ state sequence; and
$2^{lR_2^{\mathrm b}}$ ``$S_{2}$ bins'' for the $S_{2}$-sequence. Each
of the $2^{lR_0^{\mathrm b}}$ $S_{0}$-bins is indexed by a ``bin
index'' comprising $l R_0^{\mathrm b}$ bits. The index of the
$S_{0}$-bin that contains the $S_{0}$-sequence prevailing in the
current block will be conveyed to the receiver in the next
block. Likewise for the $S_{1}$-bins and the $S_{2}$-bins.

The superposition coding is used in order to create three reliable
``bit pipes.'' One bit pipe---``Encoder~1's bit pipe''---is used by
Encoder~1 to send bits it is cognizant of (e.g., fresh data, the
bits that index the bin from the previous block that
contained the previous block's $S_{1}$-sequence, and a subset of the
bits that index the bin from the previous block that contained the
previous block's $S_{0}$-sequence); another bit pipe---``Encoder~2's
bit pipe''---is similarly used by Encoder~2; and yet another bit
pipe---the ``common pipe''---is to convey bits that are known to both
encoders (e.g., all or a subset of the bits that index the bin
containing the $S_{0}$-sequence of the previous block). If the common
pipe cannot support the transmission of all the $l R_0^{\mathrm b}$
bits that index the $S_{0}$-bin, then it only conveys a subset of them
and the rest are transmitted using Encoder~1's bit pipe and/or
Encoder~2's bit pipe. Which bit pipe carries the extra bits is
determined by the vertex under consideration.

The
transmissions in future blocks and the backward decoding guarantee
that, when faced with the task of decoding a given block, the decoder
is cognizant of the indices of the bins containing the current block's
state sequences.

%
To simplify the exposition, we describe the transmission in a given
block without specifying the nature of the bits transmitted over the
bit pipes. They are merely viewed as IID random bits. Their nature
depends on the rate pair we need to convey: different rate pairs in
the capacity region will require different utilizations of the bit
pipes.  Moreover, when describing the decoding, we assume that the
indices of the bins containing the current block's state sequences 
are known to the receiver (thanks to future transmissions and the
backward decoding).
  

\subsection{A Single Block Favoring Encoder~1}
\label{sec:seq-oneblock}
We denote the length of the block $l$, the common message (comprising
the common bits) $M_{V}$, and the private messages $M_1$ and
$M_{2}$. The latter three are of respective rates $R_V^{\text{c}}$,
$R_1^{\text{p}}$, and $R_2^{\text{p}}$.  The state sequences
prevailing during the current block are denoted $s_{0}^{l}$,
$s_{1}^{l}$, and $s_{2}^{l}$. Those---or more precisely, those that
are typical---will be distributed among $2^{lR_0^{\mathrm b}}$,
$2^{lR_1^{\mathrm b}}$, and $2^{lR_2^{\mathrm b}}$ bins respectively;
here $R_0^{\mathrm b},R_1^{\mathrm b},R_2^{\mathrm b}$ are the
respective binning rates. 
%


Let
\begin{IEEEeqnarray}{c}
 \delta_l\triangleq l^{-1/4}
 \label{eq:seq-delta}
\end{IEEEeqnarray}
be the typicality parameter, so $\delta_l\downarrow0$ and
$l\,\delta_l^2=\sqrt l$.
We shall use the finite-alphabet strong-typicality bounds in
Corollary~4.8 of \cite{moser2025advanced}. Define
\begin{IEEEeqnarray}{rCl}
 p_{\min}
 \triangleq
 \min\Bigl\{
 P_{QVT_1T_2S_0S_1S_2Y}(q,v,t_1,t_2,s_0,s_1,s_2,y):\nonumber\\
 P_{QVT_1T_2S_0S_1S_2Y}(q,v,t_1,t_2,s_0,s_1,s_2,y)>0
 \Bigr\}
 \label{eq:nr-pmin}
\end{IEEEeqnarray}
and
\begin{IEEEeqnarray}{c}
 \tau_l
 \triangleq
 -2\delta_l\log p_{\min}.
 \label{eq:nr-tau-definition}
\end{IEEEeqnarray}

Since $\delta_l\downarrow0$, we have $\tau_l\to0$. The factor two is chosen so that $\tau_l$ uniformly bounds the
penalties in all conditional type-class and packing bounds
used below. 

\paragraph*{Random binning and codebooks.}
For each $k\in\{0,1,2\}$, assign every typical sequence $s_k^l$
uniformly and independently to one of $2^{l\,R_k^{\text{b}}}$ bins,
and assign atypical sequences an arbitrary default bin. Set
$L_k(s_k^l)$ to be the bin index containing $s_{k}^{l}$.

Independently of that, generate a sequence $Q^l$ IID $\sim P_Q$ and
reveal it to both encoders and to the decoder.
For every possible common message $m_V\in[1:2^{l\,R_V^{\text{c}}}]$, generate
\begin{IEEEeqnarray}{c}
 V^l(m_V)\sim\prod_{i=1}^lP_{V|Q}(\cdot\,|\,Q_i).
 \label{eq:seq-cloud-codebook}
\end{IEEEeqnarray}
For every $(m_V,m_k)$, where
$m_k\in[1:2^{l\,R_k^{\text{p}}}]$, generate conditionally and independently
\begin{IEEEeqnarray}{c}
 T_k^l(m_V,m_k)
 \sim\prod_{i=1}^lP_{T_k|VQ}
 \bigl(\cdot\,|\,V_i(m_V),Q_i\bigr),\qquad k\in\{1,2\}.
 \label{eq:seq-satellite-codebooks}
\end{IEEEeqnarray}
The Time-$i$ symbol produced by Encoder~$k$ is
\begin{IEEEeqnarray}{c}
 X_{k,i}=x_k \bigl( T_{k,i}(M_V,M_k),S_{0,i},S_{k,i}\bigr), \quad k \in \{1,2\}.
 \label{eq:seq-symbol-encoder}
\end{IEEEeqnarray}

\paragraph*{Successive decoding}
The decoder bases its decision on the indices $\ell_0 = L_0(s_0^l),\ell_1 = L_1(s_1^l),\ell_2 = L_2(s_2^l)$
of the bins containing the state sequences and on the observation
$Y^l$.  (Recall that in the Block-Markov scheme we shall employ, these
indices will be available via the backward decoding.) It uses the
following order to decode
\begin{IEEEeqnarray}{c}
 S_0^l\longrightarrow M_V\longrightarrow S_2^l\longrightarrow M_2
 \longrightarrow S_1^l\longrightarrow M_1.
 \label{eq:seq-decoding-order}
\end{IEEEeqnarray}
More precisely, it performs the following six tests.
\begin{enumerate}
\item Among all the state sequences $\hat s_0^l$ in Bin~$\ell_{0}$, it
  finds the unique $\hat s_0^l$ with $L_0(\hat s_0^l)=\ell_0$ such
  that $(Q^l,\hat s_0^l,Y^l)$ is jointly $\delta_l$-typical.
 \item Given $\hat s_0^l$, it finds the unique $\hat m_V$ such that
 $(Q^l,V^l(\hat m_V),\hat s_0^l,Y^l)$ is jointly
 $\delta_l$-typical.
 \item Given $(\hat s_0^l,\hat m_V)$, it finds the unique $\hat s_2^l$
 with $L_2(\hat s_2^l)=\ell_2$ such that
 $(Q^l,V^l(\hat m_V),\hat s_0^l,\hat s_2^l,Y^l)$ is jointly
 $\delta_l$-typical.
 \item Given $(\hat s_0^l,\hat s_2^l,\hat m_V)$, it finds the unique
 $\hat m_2$ such that
 \begin{IEEEeqnarray}{c}
  \bigl(Q^l,V^l(\hat m_V),T_2^l(\hat m_V,\hat m_2),
  \hat s_0^l,\hat s_2^l,Y^l\bigr)
  \in\mathcal T_{\delta_l}^{(l)}.
  \label{eq:seq-decoder-m2}
 \end{IEEEeqnarray}
 \item Given all previously recovered objects, it finds the unique
 $\hat s_1^l$ with $L_1(\hat s_1^l)=\ell_1$ such that
 \begin{IEEEeqnarray}{c}
  \bigl(Q^l,V^l(\hat m_V),T_2^l(\hat m_V,\hat m_2),
  \hat s_0^l,\hat s_1^l,\hat s_2^l,Y^l\bigr)
  \in\mathcal T_{\delta_l}^{(l)}.
  \label{eq:seq-decoder-s1}
 \end{IEEEeqnarray}
 \item Finally, it finds the unique $\hat m_1$ such that
 \begin{IEEEeqnarray}{c}
  \bigl(Q^l,V^l(\hat m_V),T_1^l(\hat m_V,\hat m_1),
  T_2^l(\hat m_V,\hat m_2),
  \hat s_0^l,\hat s_1^l,\hat s_2^l,Y^l\bigr)
  \in\mathcal T_{\delta_l}^{(l)}.
  \label{eq:seq-decoder-m1}
 \end{IEEEeqnarray}
\end{enumerate}
An error is declared whenever a required index is absent or nonunique. Let $P_{\mathrm e}^{(l)}$ denote the probability that at least one of
the six decisions is incorrect.

\begin{lemma}
\label{lem:seq-oneblock}
For every sufficiently large $l$, there is a deterministic single-block code
whose error probability vanishes at least as fast as $2^{-\Omega(\sqrt l)}$, provided
  \begin{IEEEeqnarray}{rCl}\label{block:lemma1}
\IEEEyesnumber
\IEEEyessubnumber*
 R_0^{\text{b}}
 &>&H(S_0\,|\,Y,Q),\label{eq:seq-rates-a}\\
 R_V^{\text{c}}
 &<&I(V;Y\,|\,S_0,Q)\quad\textnormal{or}\quad R_V^{\text{c}}=0,\label{eq:seq-rates-b}\\
 R_2^{\text{b}}
 &>&H(S_2\,|\,V,S_0,Y,Q),\label{eq:seq-rates-c}\\
 R_2^{\text{p}}
 &<&I(T_2;Y\,|\,V,S_0,S_2,Q) \quad\textnormal{or}\quad R_2^{\text{p}}=0,\label{eq:seq-rates-d}\\
 R_1^{\text{b}}
 &>&H(S_1\,|\,V,T_2,S_0,S_2,Y,Q),\label{eq:seq-rates-e}\\
 R_1^{\text{p}}
 &<&I(T_1;Y\,|\,V,T_2,S_0,S_1,S_2,Q).
 \label{eq:seq-rates-f}
\end{IEEEeqnarray}
\end{lemma}

\begin{proof}
  By symmetry, we may assume WLG that $(M_V,M_1,M_2)=(1,1,1)$. Let
  $\mathcal E_{\text{typ}}$ be the event that the transmitted tuple
\begin{IEEEeqnarray}{c}
 (Q^l,V^l(1),T_1^l(1,1),T_2^l(1,1),
 S_0^l,S_1^l,S_2^l,Y^l)
 \label{eq:seq-true-tuple}
\end{IEEEeqnarray}
is not jointly $\delta_l$-typical. Let
\begin{IEEEeqnarray}{c}
 \kappa\triangleq
 |\set{Q}||\set{V}||\set{T}_1||\set{T}_2|
 |\set{S}_0||\set{S}_1||\set{S}_2||\set{Y}|.
\end{IEEEeqnarray}
By the cardinality bounds established in Remark~\ref{rem:cardinality},
$\kappa$ is finite.

By the finite-alphabet large-deviation bound in
Corollary~4.8 of \cite{moser2025advanced},
\begin{IEEEeqnarray}{rCl}
\IEEEyesnumber
\IEEEyessubnumber*
 \Pr(\mathcal E_{\text{typ}})
 &\leq&(l+1)^\kappa
 2^{-l\,\delta_l^2\log e/(2\kappa^2)}\\
 &=&2^{-l\,E_{\text{typ},l}},
 \label{eq:seq-typical-bound}
\end{IEEEeqnarray}
where
\begin{IEEEeqnarray}{c}
 E_{\text{typ},l}
 \triangleq\frac{\delta_l^2\log e}{2\kappa^2}
 -\frac{\kappa\log(l+1)}{l},
 \label{eq:seq-typical-exponent}
\end{IEEEeqnarray}
and $l\,E_{\text{typ},l}\to\infty$ by \eqref{eq:seq-delta}.

By \eqref{eq:nr-pmin}--\eqref{eq:nr-tau-definition}, the same corollary shows that the finite-block
slack in each conditional type-class and packing bound below is at most
$\tau_l$.

We do not condition on correct decisions at preceding stages, since these
events depend on $Y^l$ and can change the channel statistics. Instead, we use
the genie-aided proof technique in
\cite[Appendix~C]{grant2001rate} and \cite[Section~6.3]{WozencraftJacobs1966}. At each stage, we assume there exists a genie that supplies the true values of all objects preceding that stage in
\eqref{eq:seq-decoding-order} to the decoder. To be specific, define
\begin{IEEEeqnarray}{rCl} \label{eq:seq-genie-events}
\IEEEyesnumber
\IEEEyessubnumber*
 \mathcal E_{S_0}
 &\triangleq&\{\text{test 1 fails}\},\\
 \mathcal G_V
 &\triangleq&\{\text{test 2 fails with the true $S_0^l$ supplied}\},
 \\
 \mathcal G_{S_2}
 &\triangleq&\{\text{test 3 fails with the true $(S_0^l,M_V)$ supplied}\},
 \\
 \mathcal G_2
 &\triangleq&\{\text{test 4 fails with the true $(S_0^l,M_V,S_2^l)$
 supplied}\},\\
 \mathcal G_{S_1}
 &\triangleq&\{\text{test 5 fails with the true
 $(S_0^l,M_V,S_2^l,M_2)$ supplied}\},\\
 \mathcal G_1
 &\triangleq&\{\text{test 6 fails with the true
 $(S_0^l,M_V,S_2^l,M_2,S_1^l)$ supplied}\}.
\end{IEEEeqnarray}
For example, in $\mathcal G_V$ the decoder is given the true $S_0^l$ and
performs the second test without requiring the first test to succeed.
All probabilities below are under the original random-coding law. We bound
each genie-aided error intersected with
$\mathcal E_{\text{typ}}^{\text c}$.

\emph{Error $\mathcal E_{S_0}$: a wrong $S_0^l$ passes the first
test.} 
For each realized $(q^l,y^l)$, the number of sequences $s_0'^l$ that are
jointly typical with $(q^l,y^l)$ is at most
\begin{IEEEeqnarray}{c}
 2^{l\,(H(S_0\,|\,Y,Q)+\tau_l)}.
 \label{eq:seq-s0-count}
\end{IEEEeqnarray}
Every such competing sequence is assigned the true bin label independently
with probability $2^{-l\,R_0^{\text{b}}}$. If $N_{S_0}$ denotes the number of
wrong competitors in that bin, then Markov's inequality gives

\begin{IEEEeqnarray}{rCl}
\IEEEyesnumber
\IEEEyessubnumber*
 \Pr(\mathcal E_{S_0}\cap\mathcal E_{\text{typ}}^{\text c})
 &\leq&\Pr(N_{S_0}\geq1,
             \mathcal E_{\text{typ}}^{\text c})\\
 &\leq&\mathbb{E}\bigl[N_{S_0}
             \mathbb{I}\{\mathcal E_{\text{typ}}^{\text c}\}\bigr]\\
 &\leq&2^{-l\,E_{S_0,l}},
 \label{eq:seq-s0-bound}
\end{IEEEeqnarray}
where
\begin{IEEEeqnarray}{c}
 E_{S_0,l}\triangleq
 R_0^{\text{b}}-H(S_0\,|\,Y,Q)-\tau_l.
 \label{eq:seq-s0-exponent}
\end{IEEEeqnarray}

\emph{Error $\mathcal G_V$: a wrong common message passes the
second test.}
Here the genie supplies the true $S_0^l$, and the decoder applies the second
test with $\hat s_0^l=S_0^l$. 
There are at most $2^{l\,R_V^{\text{c}}}$ wrong cloud center codewords. For every
$m_V'\neq1$, $V^l(m_V')$ is independent of $(S_0^l,Y^l)$ given $Q^l$.
The joint-typicality packing bound therefore gives
\begin{IEEEeqnarray}{rCl}
\IEEEyesnumber
\IEEEyessubnumber*
 &&\Pr\left(
 (Q^l,V^l(m_V'),S_0^l,Y^l)
 \in\mathcal T_{\delta_l}^{(l)}
 \right)\\
 &&\quad\leq
 2^{-l(I(V;S_0,Y\,|\,Q)-\tau_l)}\\
 &&\quad=
 2^{-l(I(V;Y\,|\,S_0,Q)-\tau_l)},
 \label{eq:seq-v-single}
\end{IEEEeqnarray}
where the equality uses $V\perp S_0$ given $Q$. Consequently,

\begin{IEEEeqnarray}{c}
 \Pr(\mathcal G_V\cap\mathcal E_{\text{typ}}^{\text c})
 \leq2^{-l\,E_{V,l}},
 \label{eq:seq-v-bound}
\end{IEEEeqnarray}
where, for notational convenience, we set
\begin{IEEEeqnarray}{c}
E_{V,l}\triangleq
\begin{cases}
I(V;Y\,|\,S_0,Q)-R_V^{\text{c}}-\tau_l,
& R_V^{\text{c}}>0,\\
+\infty, & R_V^{\text{c}}=0.
\end{cases}
\label{eq:seq-v-exponent}
\end{IEEEeqnarray}
When $R_V^{\text{c}}=0$, the cloud-center codebook is a singleton, so
$\mathcal G_V\cap\mathcal E_{\text{typ}}^{\text c}=\emptyset$.

\emph{Error $\mathcal G_{S_2}$: a wrong $S_2^l$ passes the third
test.} Here the genie supplies the true $(S_0^l,M_V)$, and the decoder applies
the third test with $\hat s_0^l=S_0^l$ and $\hat m_V=M_V$. For the genie-aided third test and fixed $(q^l,v^l,s_0^l,y^l)$, the
conditional type-class bound gives at most
\begin{IEEEeqnarray}{c}
 2^{l(H(S_2\,|\,V,S_0,Y,Q)+\tau_l)}
 \label{eq:seq-s2-count}
\end{IEEEeqnarray}
jointly typical $s_2'^l$ sequences. Each has the true bin label with
probability $2^{-l\,R_2^{\text{b}}}$. Thus

\begin{IEEEeqnarray}{c}
 \Pr(\mathcal G_{S_2}\cap\mathcal E_{\text{typ}}^{\text c})
 \leq2^{-l\,E_{S_2,l}},
 \label{eq:seq-s2-bound}
\end{IEEEeqnarray}
where
\begin{IEEEeqnarray}{c}
 E_{S_2,l}\triangleq
 R_2^{\text{b}}-H(S_2\,|\,V,S_0,Y,Q)-\tau_l.
 \label{eq:seq-s2-exponent}
\end{IEEEeqnarray}

\emph{Error $\mathcal G_2$: a wrong Encoder~2 satellite message
passes the fourth test.} Here the genie supplies the true $(S_0^l,M_V,S_2^l)$, and the decoder
applies the fourth test with $\hat s_0^l=S_0^l$, $\hat m_V=M_V$, and
$\hat s_2^l=S_2^l$.
For each $m_2'\neq1$, the codeword $T_2^l(1,m_2')$ is independent of
$Y^l$ given $(Q^l,V^l(1),S_0^l,S_2^l)$. Hence
\begin{IEEEeqnarray}{c}
\Pr\left(
 (Q^l,V^l(1),T_2^l(1,m_2'),S_0^l,S_2^l,Y^l)
 \in\mathcal T_{\delta_l}^{(l)}
 \right)\leq
 2^{-l(I(T_2;Y\,|\,V,S_0,S_2,Q)-\tau_l)}.
 \label{eq:seq-m2-single}
\end{IEEEeqnarray}
There are at most $2^{l\,R_2^{\text{p}}}$ competitors, and therefore

\begin{IEEEeqnarray}{c}
 \Pr(\mathcal G_2\cap\mathcal E_{\text{typ}}^{\text c})
 \leq2^{-l\,E_{2,l}},
 \label{eq:seq-m2-bound}
\end{IEEEeqnarray}
where
\begin{IEEEeqnarray}{c}
E_{2,l}\triangleq
\begin{cases}
I(T_2;Y\,|\,V,S_0,S_2,Q)-R_2^{\text{p}}-\tau_l,
& R_2^{\text{p}}>0,\\
+\infty, & R_2^{\text{p}}=0.
\end{cases}
\label{eq:seq-m2-exponent}
\end{IEEEeqnarray}
When $R_2^{\text{p}}=0$, the Encoder~2 satellite codebook is a
singleton, so
$\mathcal G_2\cap\mathcal E_{\text{typ}}^{\text c}=\emptyset$.

\emph{Error $\mathcal G_{S_1}$: a wrong $S_1^l$ passes the fifth
test.} Here the genie supplies the true $(S_0^l,M_V,S_2^l,M_2)$, and the decoder
applies the fifth test with $\hat s_0^l=S_0^l$, $\hat m_V=M_V$,
$\hat s_2^l=S_2^l$, and $\hat m_2=M_2$.
For the genie-aided fifth test and every realized
$(q^l,v^l,t_2^l,s_0^l,s_2^l,y^l)$, at most
\begin{IEEEeqnarray}{c}
 2^{l(H(S_1\,|\,V,T_2,S_0,S_2,Y,Q)+\tau_l)}
 \label{eq:seq-s1-count}
\end{IEEEeqnarray}
state sequences pass the typicality test. Random binning places each wrong
sequence in the true bin with probability $2^{-l\,R_1^{\text{b}}}$, so

\begin{IEEEeqnarray}{c}
 \Pr(\mathcal G_{S_1}\cap\mathcal E_{\text{typ}}^{\text c})
 \leq2^{-l\,E_{S_1,l}},
 \label{eq:seq-s1-bound}
\end{IEEEeqnarray}
where
\begin{IEEEeqnarray}{c}
 E_{S_1,l}\triangleq
 R_1^{\text{b}}-H(S_1\,|\,V,T_2,S_0,S_2,Y,Q)-\tau_l.
 \label{eq:seq-s1-exponent}
\end{IEEEeqnarray}

\emph{Error $\mathcal G_1$: a wrong Encoder~1 satellite message
passes the last test.} Here the genie supplies the true $(S_0^l,M_V,S_2^l,M_2,S_1^l)$, and the
decoder applies the last test with $\hat s_0^l=S_0^l$, $\hat m_V=M_V$,
$\hat s_2^l=S_2^l$, $\hat m_2=M_2$, and $\hat s_1^l=S_1^l$.
For every $m_1'\neq1$, $T_1^l(1,m_1')$ is independent of $Y^l$ given
$(Q^l,V^l(1),T_2^l(1,1),S_0^l,S_1^l,S_2^l)$. Thus
\begin{IEEEeqnarray}{rCl}
 \Pr\Bigl(
 (Q^l,V^l(1),&&T_1^l(1,m_1'),T_2^l(1,1),
 S_0^l,S_1^l,S_2^l,Y^l)
 \in\mathcal T_{\delta_l}^{(l)}
 \Bigr)\nonumber\\
 &&\quad\leq
 2^{-l(I(T_1;Y\,|\,V,T_2,S_0,S_1,S_2,Q)-\tau_l)}.
 \label{eq:seq-m1-single}
\end{IEEEeqnarray}
After a union bound over at most $2^{l\,R_1^{\text{p}}}$ competitors,

\begin{IEEEeqnarray}{c}
 \Pr(\mathcal G_1\cap\mathcal E_{\text{typ}}^{\text c})
 \leq2^{-l\,E_{1,l}},
 \label{eq:seq-m1-bound}
\end{IEEEeqnarray}
where
\begin{IEEEeqnarray}{c}
 E_{1,l}\triangleq
 I(T_1;Y\,|\,V,T_2,S_0,S_1,S_2,Q)-R_1^{\text{p}}-\tau_l.
 \label{eq:seq-m1-exponent}
\end{IEEEeqnarray}

Let

\begin{IEEEeqnarray}{c}
 E_l\triangleq
 \min\{E_{S_0,l},E_{V,l},E_{S_2,l},E_{2,l},E_{S_1,l},E_{1,l}\}.
 \label{eq:seq-min-exponent}
\end{IEEEeqnarray}
If the actual decoder makes an error on
$\mathcal E_{\text{typ}}^{\text c}$, consider its first failing test.
All preceding decisions are then correct. Therefore, at this test, the
actual decoder uses exactly the same preceding objects as those supplied
by the genie in the corresponding genie-aided test. Consequently, the
occurrence of this first failure implies the occurrence of the
corresponding genie-aided error event. Hence,
\begin{IEEEeqnarray}{c}
\{\text{actual decoder errors}\}
 \cap\mathcal E_{\text{typ}}^{\text c}
 \subseteq
 \bigl(\mathcal E_{S_0}\cup\mathcal G_V\cup\mathcal G_{S_2}
 \cup\mathcal G_2\cup\mathcal G_{S_1}\cup\mathcal G_1\bigr)
 \cap\mathcal E_{\text{typ}}^{\text c}.
 \label{eq:seq-first-error-genie}
\end{IEEEeqnarray}
Therefore,
\begin{IEEEeqnarray}{rCl}
\IEEEyesnumber
\IEEEyessubnumber*
 P_{\text e}^{(l)}
 &\leq&\Pr(\mathcal E_{\text{typ}})
 +\Pr(\mathcal E_{S_0}\cap\mathcal E_{\text{typ}}^{\text c})
 +\Pr(\mathcal G_V\cap\mathcal E_{\text{typ}}^{\text c})
 +\Pr(\mathcal G_{S_2}\cap\mathcal E_{\text{typ}}^{\text c})
 \nonumber\\
 &&{}+\Pr(\mathcal G_2\cap\mathcal E_{\text{typ}}^{\text c})
 +\Pr(\mathcal G_{S_1}\cap\mathcal E_{\text{typ}}^{\text c})
 +\Pr(\mathcal G_1\cap\mathcal E_{\text{typ}}^{\text c})\\
 &\leq&2^{-l\,E_{\text{typ},l}}+6\cdot2^{-l\,E_l}.
 \label{eq:seq-oneblock-error}
\end{IEEEeqnarray}
Since $\tau_l\to0$, each strict rate condition in
\eqref{eq:seq-rates-a}--\eqref{eq:seq-rates-f} makes the corresponding
error exponent positive for all sufficiently large $l$.
Together with $l\,E_{\mathrm{typ},l}\to\infty$, this implies
$P_{\mathrm e}^{(l)}\to0$. Averaging over
the random construction gives the existence of a deterministic one-block code with the same
vanishing-error property.
\end{proof}

The one-block lemma specifies only the binning and bit-pipe rates
supported by the successive decoder. The message-rate allocation to
fresh messages and previous-block state labels is postponed to the
interblock construction.

\begin{remark}\label{rem:red_favoring}
  We can also favor Encoder~2 by reversing the decoding order
of~\eqref{eq:seq-decoding-order} to have it be 
\begin{IEEEeqnarray}{c}
 S_0^l \longrightarrow M_V \longrightarrow S_1^l\longrightarrow M_1
 \longrightarrow S_2^l\longrightarrow M_2.
 \label{eq:red_seq-decoding-order}
\end{IEEEeqnarray}
This will demonstrate achievability whenever
\begin{IEEEeqnarray}{rCl}\label{block:red_lemma1}
\IEEEyesnumber
\IEEEyessubnumber*
R_0^{\text{b}}
&>&H(S_0\,|\,Y,Q),\label{eq:seq-rates-enc2-a}\\
R_V^{\text{c}}
&<&I(V;Y\,|\,S_0,Q)\quad\textnormal{or}\quad R_V^{\text{c}}=0,\label{eq:seq-rates-enc2-b}\\
R_1^{\text{b}}
&>&H(S_1\,|\,V,S_0,Y,Q),\label{eq:seq-rates-enc2-c}\\
R_1^{\text{p}}
&<&I(T_1;Y\,|\,V,S_0,S_1,Q) \quad\textnormal{or}\quad R_1^{\text{p}}=0,\label{eq:seq-rates-enc2-d}\\
R_2^{\text{b}}
&>&H(S_2\,|\,V,T_1,S_0,S_1,Y,Q),\label{eq:seq-rates-enc2-e}\\
R_2^{\text{p}}
&<&I(T_2;Y\,|\,V,T_1,S_0,S_1,S_2,Q).
\label{eq:seq-rates-enc2-f}
\end{IEEEeqnarray}
\end{remark}
Note that Inequalities~\eqref{block:red_lemma1} are obtained from
  those of \eqref{block:lemma1} by swapping the index~1 and the
  index~2 in all quantities of interest (namely, the rate indices, the
  state indices, and the strategy indices).

\subsection{Rate selection and interblock allocation}

We next describe the rate allocation that allows us to approach the
vertices of the polytope; see Section~\ref{sec:seq-net-rates}.
To guarantee that the required inequalities hold,
we choose
the rates $R_0^{\text{b}}$, $R_1^{\text{b}}$, $R_1^{\text{p}}$,
$R_2^{\text{b}}$, $R_2^{\text{p}}$, in Lemma~\ref{lem:seq-oneblock},
for a sufficiently small $\epsilon>0$, as
\begin{subequations}
  \label{block:amos_aaa}
\begin{IEEEeqnarray}{rCl}
\IEEEyesnumber*
R_0^{\text{b}}
&=&H(S_0\,|\,Y,Q)+\epsilon,\label{eq:seq-closed-rates-a}\\
R_2^{\text{b}}
&=&H(S_2\,|\,V,S_0,Y,Q)+\epsilon,\\
R_2^{\text{p}}
&=&\max\bigl\{I(T_2;Y\,|\,V,S_0,S_2,Q)-\epsilon,0\bigr\},\\
R_1^{\text{b}}
&=&H(S_1\,|\,V,T_2,S_0,S_2,Y,Q)+\epsilon,\\
R_1^{\text{p}}
&=&I(T_1;Y\,|\,V,T_2,S_0,S_1,S_2,Q)-\epsilon.
\label{eq:seq-closed-rates-e}
\end{IEEEeqnarray}
\end{subequations}
As to rate $R_V^{\text{c}}$ of the common pipe, we note
that~\eqref{eq:seq-rates-b} requires either
$R_V^{\text{c}}<I(V;Y\,|\,S_0,Q)$ or $R_V^{\text{c}}=0$, so we do not
want it to exceed $\max\{I(V;Y\,|\,S_0,Q)-\epsilon,0\}$. And since it
is only used to convey all (or some) of the bits that index the
$S_{0}$-bin, we cap it at $R_0^{\text{b}}$. We 
therefore choose
\begin{IEEEeqnarray}{c}
R_V^{\text{c}}
=\min\Bigl\{R_0^{\text{b}},
\max\{I(V;Y\,|\,S_0,Q)-\epsilon,0\}\Bigr\}.
\label{eq:seq-common-rate}
\end{IEEEeqnarray}
(Recall that if indexing the $S_{0}$-bin requires more bits than the common message can
carry, the remaining bits are sent over the other bit pipes.)

We next show how the bit pipes are used to approach the different
vertices of the polytope.

\subsubsection{Blue vertex}\label{sec:blue_achie}
We use the Encoder~1-favoring decoding order of
Section~\ref{sec:seq-oneblock}. The common message carries
$l\,R_V^{\text{c}}$ of the $l\, R_{0}^{\text{b}}$ bits that index the
$S_{0}$-bin. Encoder~1's bit pipe 
carries Encoder~1's fresh message and the preceding
$l\,R_1^{\text{b}}$ bits that index the $S_{1}$-bin. Encoder~2's
bit pipe 
carries its fresh message, the preceding $l\,R_2^{\text{b}}$ bits that index
the $S_2$-bin, and the remaining $l(R_0^{\text{b}}-R_V^{\text{c}})$
bits that couldn't be carried by the common bit pipe and that are still
needed to complete the description of the $S_{0}$-bin. 
By~\eqref{block:amos_aaa}, this
allocation allows for the transmission rate
\begin{IEEEeqnarray}{rCl}
 \label{eq:seq-blue-first}
\IEEEyesnumber
\IEEEyessubnumber*
R_1 &=&  R_1^{\text{p}}-R_1^{\text{b}}
 =I(T_1;Y\,|\,V,T_2,S_0,S_1,S_2,Q)     -H(S_1\,|\,V,T_2,S_0,S_2,Y,Q)-2\epsilon\\
 &=&A_1-2\epsilon.
\end{IEEEeqnarray}
As to $R_{2}$, the above allocation results in
\begin{IEEEeqnarray}{rCl}
  R_{2} & = & R_2^{\text{p}}
-\bigl(R_2^{\text{b}}+(R_0^{\text{b}}-R_V^{\text{c}})\bigr).
\end{IEEEeqnarray}

Using the fact that---by the definitions
in~\eqref{eq:seq-region-parameters}, the chain rule, and the
independence of $S_{0}$ and $Q$--- 
\begin{IEEEeqnarray}{rCl}
\IEEEyesnumber
\IEEEyessubnumber*
 C_{\rm p}-A_1
 &=&I(T_2,S_2;Y\,|\,V,S_0,Q)-H(S_2)\\
 C_{\rm t}-A_1
 &=&I(V,T_2,S_2;Y\,|\,S_0,Q)
     -H(S_0\,|\,Y,Q)-H(S_2)
 \label{eq:seq-blue-residuals}
\end{IEEEeqnarray}
we obtain from~\eqref{block:amos_aaa}--\eqref{eq:seq-common-rate} that
\begin{IEEEeqnarray}{rCl}
\lim_{\epsilon \downarrow 0} \Bigl( R_2^{\text{p}}
-\bigl(R_2^{\text{b}}+(R_0^{\text{b}}-R_V^{\text{c}})\bigr) \Bigr)
& = &
\min\{C_{\rm p}-A_1,C_{\rm t}-A_1\}
=C-A_1.
\label{eq:seq-blue-second}
\end{IEEEeqnarray}
As $C>A_1$ (blue vertex), the right-hand side is strictly
positive. Hence, for all sufficiently small $\epsilon>0$, both
fresh-message rates are strictly positive.  Therefore, our scheme produces a
rate pair approaching
\begin{IEEEeqnarray}{c}
  (A_1,C-A_1).
  \label{eq:noise123}
\end{IEEEeqnarray}

\subsubsection{Red vertex}\label{sec:red_achie}
The red vertices occur when $C \leq A_1$, which is thus what we now
assume. They have the form $(C,0)$. The single-block scheme we propose
to approach this vertex is a time-sharing scheme between the Encoder~1-favoring scheme of Lemma~\ref{lem:seq-oneblock} and the Encoder~2-favoring scheme of Remark~\ref{rem:red_favoring}.

The formal calculation we performed to approach the blue vertex using
the Encoder~1-favoring decoding order yields (as in~\eqref{eq:noise123})
the ``bookkeeping'' pair
\begin{subequations}
  \begin{IEEEeqnarray}{rCl}
    \mathbf r^{(1)} & =& (A_1,C-A_1).
  \end{IEEEeqnarray}

  We can also perform a similar formal calculation for the
  Encoder~2-favoring decoding order: Recall that
  Inequalities~\eqref{block:red_lemma1} are obtained from those of
  \eqref{block:lemma1} by swapping the index~1 and the index~2 in all
  quantities of interest (namely, the rate indices, the state indices,
  and the strategy indices). This also results in the swapping of
  $A_{1}$ with $A_{2}$. Thus, the formal calculation yields the
  bookkeeping pair
  \begin{IEEEeqnarray}{rCl}
    \mathbf r^{(2)}& = &(C-A_2,A_2).
  \end{IEEEeqnarray}
\end{subequations}

We used ``formal,'' and ``bookkeeping pair'' (and not ``rate pair'')
because some of the components of $\mathbf r^{(1)}$ and
$\mathbf r^{(2)}$ may be negative. Negative components merely mean
that the bitpipes cannot support the transmission of the required part
of the bin-index. It thus only represents a possible deficit after
subtracting the bin-index bits; it does not represent a negative
physical transmission rate. This is not an issue as long as the
components that result from time-sharing are nonnegative, i.e., as
long as the bitpipes we obtain after time sharing can support the
transmission of the bin indices.

Since
$\set R_{<}(P)$ is nonempty, we have $A_2>0$. Together with
$C\leq A_1$, this implies that $A_1+A_2-C>0$.

For sufficiently small $\epsilon>0$, choose the duty cycle
\begin{IEEEeqnarray}{c}
\alpha_\epsilon
\triangleq
\frac{A_2}{A_1+A_2-C}-\sqrt{\epsilon}.
\label{eq:seq-red-weight}
\end{IEEEeqnarray}
With this choice, $0<\alpha_\epsilon<1$ and
\begin{IEEEeqnarray}{rCl}
&&\alpha_\epsilon\mathbf r^{(1)}
+(1-\alpha_\epsilon)\mathbf r^{(2)}
\nonumber\\
&=&
\bigl(
C-(A_1+A_2-C)\sqrt{\epsilon},
(A_1+A_2-C)\sqrt{\epsilon}
\bigr),
\label{eq:seq-red-combination}
\end{IEEEeqnarray}
which converges to $(C,0)$ as $\epsilon \downarrow 0$.

Time-sharing can be implemented 
by dividing the length-$l$ block into two subblocks
of lengths
\begin{IEEEeqnarray}{c}
l_1=\lfloor\alpha_\epsilon l\rfloor,
\qquad l_2=l-l_1.
\label{eq:seq-red-lengths}
\end{IEEEeqnarray}
In the first subblock we employ the Encoder~1-favoring decoding order,
and in the second the Encoder~2-favoring order. The rates in
the second subblock are the symmetric counterparts of
\eqref{block:amos_aaa}--\eqref{eq:seq-common-rate}, i.e., the rates
obtained by swapping the index 1 and the index 2 in all the variables.
%
%
Let $(R_0^{\text{b},(q)},R_1^{\text{b},(q)},
R_2^{\text{b},(q)},R_V^{\text{c},(q)},
R_1^{\text{p},(q)},R_2^{\text{p},(q)})$
denote the rates in the subblock~$q$, where $q \in \{1,2\}$.

In each subblock, the common pipe carries as many bits of the
preceding block's $S_0$-bin index as its rate permits.  Some of the
remaining $S_0$-bin-index bits (if any) are assigned
to Encoder~2's bit pipe in the Encoder~1-favoring subblock (the first
subblock), and the
rest are assigned to Encoder~1's bit pipe in the Encoder~2-favoring
subblock (the second subblock). 

In each subblock, each encoder uses its bit-pipe rate to carry its
fresh message and all bin-index bits assigned to it:
its own state bin-index bits and possibly some of the remaining
$S_0$-bin-index bits. Consequently, the fresh-message rates equal the
total bit-pipe rates minus the total assigned bin-index rates:
\begin{IEEEeqnarray}{rCl}
 \label{eq:seq-red-subtraction}
\IEEEyesnumber
\IEEEyessubnumber*
 R_1^{(l)}
 &=&\frac{1}{l}\Bigl\{
 l_1(R_1^{\text{p},(1)}-R_1^{\text{b},(1)})
 +l_2\bigl[R_1^{\text{p},(2)}-R_1^{\text{b},(2)}
             -(R_0^{\text{b},(2)}-R_V^{\text{c},(2)})\bigr]
 \Bigr\},\\
 R_2^{(l)}
 &=&\frac{1}{l}\Bigl\{
 l_1\bigl[R_2^{\text{p},(1)}-R_2^{\text{b},(1)}
             -(R_0^{\text{b},(1)}-R_V^{\text{c},(1)})\bigr]
 +l_2(R_2^{\text{p},(2)}-R_2^{\text{b},(2)})
 \Bigr\}.
\end{IEEEeqnarray}

Having described the time-sharing scheme, we can return to the
limiting rate: The positive second coordinate in
\eqref{eq:seq-red-combination} is of order $\sqrt{\epsilon}$, whereas
the rate loss caused by the coding slack is of order $\epsilon$.
Consequently, for all sufficiently small fixed $\epsilon>0$ and all
sufficiently large $l$, both fresh-message rates (especially
$R_2^{(l)}$) in \eqref{eq:seq-red-subtraction} are strictly
positive. Moreover,
\begin{IEEEeqnarray}{c}
\bigl(R_1^{(l)},R_2^{(l)}\bigr)\longrightarrow(C,0),
\end{IEEEeqnarray}
where the limit is taken first as $l\to\infty$ for fixed $\epsilon$
and then as $\epsilon\downarrow0$.


\subsection{The Terminating Block and Backward Decoding}
\label{sec:seq-compound}

To conclude the achievability proof, we now describe the last block
and the backward decoding.

In the last block no fresh data are sent. Its purpose is to describe
the bins containing the state sequences that prevailed during the
previous (next-to-last) block. Recall that in the next-to-last block
(as in all previous blocks) bins were generated for the different
state sequences. Those are used in the last block. More specifically,
Encoder~1 conveys to the decoder the index of the $S_{1}$-bin that
contains the $S_{1}$-state sequence that prevailed in the next-to-last
block. Likewise, Encoder~2 conveys the index of the $S_{2}$-bin, and
the two encoders can cooperate in sending the index of the
$S_{0}$-bin. The length $l'$ of the last block is chosen so as to
allow this transmission (with causal state information).
Assumption~\ref{ass:nontri} guarantees that $l'$ is within a constant
factor of $l$, i.e., $l'(l)=O(l)$.

Not sending data in the last block incurs a rate-loss that can be made
arbitrarily small by having the number of blocks large. Also, the states
prevailing in the last block are not conveyed to the receiver, and
they thus contribute to the symbol error-rate in the state's
reconstruction. Again, this contribution to the symbol error-rate is
arbitrarily small if the number of blocks is large.

Once the last block has terminated, the decoder can recover the
indices of the bins containing the state sequences that prevailed
during the next-to-last block and can thus proceed to decoding the
next-to-last transmission. This backward decoding continues until the
first block is decoded.

Note that the coding slack $\epsilon$ used above is distinct from the
$\epsilon_0$ in the definition of achievability. Given any neighborhood
$\mathcal N$ of either endpoint, choose $\epsilon$ sufficiently small
and let the blocklength and the number of blocks be large enough so that
the construction above yields a strictly positive rate pair
$(R_1,R_2)\in\mathcal N$ with vanishing message and
state-reconstruction errors. Since $\mathcal N$ is open, there exists
$\epsilon_0>0$ such that
\begin{IEEEeqnarray}{c}
    (R_1-\epsilon_0,R_2-\epsilon_0) \in \mathcal N\cap\mathbb R_{\geq0}^2.
\end{IEEEeqnarray}
So $(R_1-\epsilon_0,R_2-\epsilon_0)$ is
achievable by definition. Thus, every neighborhood of either endpoint
contains an achievable rate pair, which demonstrates that the vertex belongs to the
closure of the achievable rate pairs, namely, the capacity region.

\section*{Acknowledgment}
This work was supported by the Swiss National Science Foundation (SNSF) under Grant 200021-
215090.

\appendix
\section{Deferred Proofs}
\label{app:deferred-proofs}

\subsection{Proof of Remark~\ref{rem:common_state}}
\label{app:proof_common_state}

\begin{IEEEproof}[Proof of Remark~\ref{rem:common_state}]
  To prove the equivalence, it suffices to show that
  if one encoder cannot transmit at positive rate using Shannon
  strategies, then nor can it using general causal strategies. The
  statement will then follow using a time-sharing argument. Suppose, then, that $\set{S}_1 = \set{S}_2 = \emptyset$ and that every symbol in $\set{S}_0$
has positive probability under $P_{S_0}$. Without loss of
  generality, assume that it is Encoder~1 that cannot transmit
  reliably at any positive rate using Shannon strategies. 

This is
  equivalent to
    \begin{IEEEeqnarray}{c}\label{eq:shannon_inva}
        \tilde{W}(\cdot\,|\,t_1,t_2) = \tilde{W}(\cdot\,|\,t_1',t_2),\quad \forall t_1,t_1'\in\set{X}_1^{\set{S}_0},\,t_2\in\set{X}_2^{\set{S}_0},
    \end{IEEEeqnarray}
    where $\tilde{W}$ denotes the channel induced by treating the Shannon
    strategies as inputs
    \begin{IEEEeqnarray}{rCl}
      \tilde{W}(\cdot\,|\,t_1,t_2) & = & \sum_{s_0} W\bigl(\cdot\,|\,t_1(s_0),t_2(s_0),s_0\bigr)P_{S_0}(s_0)
    \end{IEEEeqnarray}
    and $\set{X}_1^{\set{S}_0}$ and $\set{X}_2^{\set{S}_0}$ denote the
    sets of Shannon strategies for Encoders~1 and~2, respectively.

    As we next argue, \eqref{eq:shannon_inva} implies that
    \begin{IEEEeqnarray}{c}\label{eq:symbol_inva}
        W(\cdot\,|\,x_1,x_2,s_0) = W(\cdot\,|\,x_1',x_2,s_0),\quad \forall x_1,x_1'\in\set{X}_1,\,x_2\in\set{X}_2\,,s_0\in\set S_0.
    \end{IEEEeqnarray}
    Indeed, if~\eqref{eq:symbol_inva} fails for some tuple
    $(\bar{x}_1,\bar{x}_1',\bar{x}_2,\bar{s}_0)$, then the strategies
    $(t_1,t_1',t_2)$ defined as
    \begin{IEEEeqnarray}{rCl}
    \IEEEyesnumber
    \IEEEyessubnumber*
    t_1(s_0) &=& \bar{x}_1,\;\text{for all $s_0$.}\\
        t'_1(s_0) &=& \begin{cases}
                   \bar{x}'_1 & s_0=\bar{s}_0\\
                   \bar{x}_1 & s_0\neq\bar{s}_0.
        \end{cases}\\
        t_2(s_0) &=& \bar{x}_2,\;\text{for all $s_0$.}
    \end{IEEEeqnarray}
    violate \eqref{eq:shannon_inva}.

    We next argue that~\eqref{eq:symbol_inva} implies that no positive
    rate tuple is achievable with any causal strategy (be it of the
    Shannon kind or not).
    
     To that end, consider any blocklength-$l$ causal strategies
     $g_1^l = \{g_{1,i} \colon \set{S}_0^{\times i}\to\set{X}_1\}_{i=1}^l$,
     $g_1^{'l} = \{g'_{1,i} \colon \set{S}_0^{\times i}\to\set{X}_1\}_{i=1}^l$, and
     $g_2^l = \{g_{2,i} \colon \set{S}_0^{\times i}\to\set{X}_2\}_{i=1}^l$.
     We claim that $(g_1^l,g_2^l)$ and $(g_1^{'l},g_2^l)$ induce identical
     distributions on $\set{Y}^{\times l}$. Indeed,
     \begin{IEEEeqnarray}{rCl}
     \IEEEyesnumber
     \IEEEyessubnumber*
         \text{Pr}\bigl\{Y^l = y^l\,|\, g_1^l, g_2^l\bigr\}& =&
         \sum_{s_0^l\in\set{S}_0^{\times l}}P_{S_0^l}(s_0^l) \, W^l\Bigl(y^l\,\Big|\,g_1^l(s_0^l),g_2^l(s_0^l),s_0^l\Bigr)\label{eq:prod_invaa}\\
         &=& \sum_{s_0^l\in\set{S}_0^{\times l}}P_{S_0^l}(s_0^l) \, W^l\Bigl(y^l\,\Big|\,g_1^{'l}(s_0^l),g_2^l(s_0^l),s_0^l\Bigr)\label{eq:prod_invab}\\
         &=&  \text{Pr}\bigl\{Y^l = y^l\,|\, g_1^{'l}, g_2^l\bigr\}\label{eq:prod_invac},
     \end{IEEEeqnarray}
     where \eqref{eq:prod_invab} follows by applying
     \eqref{eq:symbol_inva} to the $l$-fold product channel. We can now conclude that $R_1$ must be zero also for causal strategies.

    The same proof extends to general bipartite networks in the sense of
    \cite{lapidoth2026equivalence}, as well as to non IID state sequences.
    
    To prove the second statement of the remark, it suffices to
    provide an example. Let
    $\set{X}_1=\set{X}_2=\set{S}_0=\set{S}_2=\{0,1\}$, let
    $\set{S}_1=\emptyset$, and let $\set{Y}=\{0,1\} \times \{0,1\}$. The
    random variables $S_0$ and $S_2$ are independent
    $\operatorname{Ber}(0.5)$ random bits, and
     \begin{IEEEeqnarray}{rCl}
         Y&=&\begin{cases}
             (X_1\oplus S_2,S_0), & S_0=1,\\
             (X_2,S_0), & S_0=0.
         \end{cases}
     \end{IEEEeqnarray}
     To see that no $R_{1} > 0$ is achievable using Shannon
     strategies, consider any strategies  $t_1\colon \set{S}_0\to\set{X}_1$
     and $t_2 \colon \set{S}_0\times\set{S}_2\to\set{X}_2$ (where we have
     dropped the dependence of $t_{1}$ on $S_{1}$ because the latter
     is deterministic). When $S_0=1$, the output $t_1(1)\oplus S_2$ is
     uniform irrespective of $t_1$. When $S_0=0$, the output is
     $(t_2(0,S_2),0)$ and hence, also in this case, does not depend on
     $t_1$. Consequently, the induced Shannon-strategy channel
     satisfies
     \begin{IEEEeqnarray}{c}
         P_{Y|T_1T_2}(\cdot\,|\,t_1,t_2)
         =P_{Y|T_1T_2}(\cdot\,|\,t_1',t_2)
     \end{IEEEeqnarray}
     for every $t_1,t_1'$, and therefore every Shannon-strategy achievable
     rate pair has $R_1=0$.

     We next provide a causal coding scheme (that is not based on
     Shannon strategies) that allows us to communicate reliably at
     positive rates, in fact, at rates $(1/4,1/4)$. Our scheme is
     based on the observation that when $S_{0} = 1$ the reason the
     receiver cannot recover Encoder~1's input symbol is that the
     state $S_{2}$ masks it. Had $S_{2}$ been conveyed to the
     receiver, the latter could have subtracted it and recovered
     Encoder~1's symbol. In our scheme, Encoder~2 uses some of its
     throughput to convey $S_{2}$ to the receiver to facilitate this.
     
     For a fixed $0<\delta<1/4$, our scheme divides each sufficiently
     long blocklength-$n$ into two equal-length phases. During the
     first interval, Encoder~1 sends fresh information at the first
     $\lfloor n(1/4-\delta)\rfloor$ time instances at which
     $S_{0} = 1$ (with a block error if there are fewer than required
     such time instances). Note that, since $S_{0}$ is part of the
     channel output, these time instances are known to the
     receiver. Moreover, since $S_{0}$ is the common state, Encoder~2
     is also cognizant of these time instances. It can therefore
     record the values of $S_{2}$ at these times. (Encoder~2 will
     convey these values to the receiver in the second phase.)

     Also in the first phase, Encoder~2 sends its own fresh information in
     the first $\lfloor n(1/4-\delta)\rfloor$ positions at which
     $S_0=0$.

     In the second phase Encoder~2 uses the first
     $\lfloor n(1/4-\delta)\rfloor$ time instances at which $S_0=0$ to
     transmit the values of $S_{2}$ that it recorded in the first
     phase. All unused inputs are fixed dummy symbols.  The decoder
     observes $S_0$ as the second component of $Y$, decodes
     Encoder~2's bits directly, and uses the forwarded $S_2$ values to
     remove the masks from Encoder~1's stored outputs. By the law of
     large numbers, all three required collections of positions have
     the prescribed size with probability tending to one. Hence
     $(1/4-\delta,1/4-\delta)$ is achievable for every $\delta>0$, and
     closure gives $(R_1,R_2)=(1/4,1/4)$.
\end{IEEEproof}

\subsection{Proof of Remark~\ref{rem:cardinality}}
\label{app:proof-cardinality}

\begin{IEEEproof}
We apply Carath\'eodory's theorem in the form stated in
\cite[Theorem~1.31]{moser2025advanced}. Fix $q\in\set{Q}$ and an arbitrary
$y_\star\in\set{Y}$. For each $v$, consider the vector consisting of
\begin{IEEEeqnarray}{rCl}
    &&\bigl(P_{Y|VQ}(y\,|\,v,q):
        y\in\set{Y}\setminus\{y_\star\}\bigr),\nonumber\\
    &&H(Y\,|\,X_1,X_2,S_0,S_1,S_2,V=v,Q=q),\nonumber\\
    &&I(X_1,S_1;Y\,|\,X_2,S_0,S_2,V=v,Q=q),\nonumber\\
    &&I(X_2,S_2;Y\,|\,X_1,S_0,S_1,V=v,Q=q),\nonumber\\
    &&I(X_1,X_2,S_1,S_2;Y\,|\,S_0,V=v,Q=q).
    \label{eq:v-support-functionals}
\end{IEEEeqnarray}
These are $|\set{Y}|+3$ functions of $v$. Their average under $P_{V|Q}(\cdot\,|\,q)$ can
therefore be preserved using at most $|\set{Y}|+4$ values of $v$.

The last three entries of \eqref{eq:v-support-functionals} preserve the first
three mutual informations in \eqref{eq:region}. The first
$|\set{Y}|-1$ entries preserve $P_{Y|Q}(\cdot\,|\,q)$ and hence
$H(Y\,|\,Q=q)$, while the remaining entropy entry preserves
\begin{IEEEeqnarray}{rCl}
&&H(Y\,|\,X_1,X_2,S_0,S_1,S_2,Q=q)\nonumber\\
&&\quad=\sum_vP_{V|Q}(v\,|\,q)
H(Y\,|\,X_1,X_2,S_0,S_1,S_2,V=v,Q=q).
\end{IEEEeqnarray}
It follows from
\begin{IEEEeqnarray}{rCl}
&&I(X_1,X_2,S_0,S_1,S_2;Y\,|\,Q=q)\nonumber\\
&&\quad=H(Y\,|\,Q=q)
-H(Y\,|\,X_1,X_2,S_0,S_1,S_2,Q=q)
\end{IEEEeqnarray}
that the fourth mutual information in \eqref{eq:region} is preserved as well.
The retained support points continue to define input kernels of the required
product form in \eqref{eq:pmf}.

After this reduction, associate with each $q$ the vector in $\mathbb R^4$
whose entries are the four mutual informations on the right-hand side of
\eqref{eq:region}, before subtracting the state entropies. Their averages under
$P_Q$ are precisely the four mutual informations in \eqref{eq:region}.
A second application of Carath\'eodory's theorem preserves this average using
at most five values of~$q$. 
\end{IEEEproof}


\bibliographystyle{plain}  
\bibliography{./bibliofile}

\begin{thebibliography}{1}
\bibliographystyle{IEEEtran}

\bibitem{ref1}
{\it{Mathematics Into Type}}. American Mathematical Society. [Online]. Available: https://www.ams.org/arc/styleguide/mit-2.pdf

\bibitem{ref2}

\bibitem{ref3}
F. Mittelbach and M. Goossens, {\it{The \LaTeX Companion}}, 2nd ed. Boston, MA, USA: Pearson, 2004.

\bibitem{ref4}
G. Gr\"atzer, {\it{More Math Into LaTeX}}, New York, NY, USA: Springer, 2007.

\bibitem{ref5}M. Letourneau and J. W. Sharp, {\it{AMS-StyleGuide-online.pdf,}} American Mathematical Society, Providence, RI, USA, [Online]. Available: http://www.ams.org/arc/styleguide/index.html

\bibitem{ref6}
H. Sira-Ramirez, ``On the sliding mode control of nonlinear systems,'' \textit{Syst. Control Lett.}, vol. 19, pp. 303--312, 1992.

\bibitem{ref7}
A. Levant, ``Exact differentiation of signals with unbounded higher derivatives,''  in \textit{Proc. 45th IEEE Conf. Decis.
Control}, San Diego, CA, USA, 2006, pp. 5585--5590. DOI: 10.1109/CDC.2006.377165.

\bibitem{ref8}
M. Fliess, C. Join, and H. Sira-Ramirez, ``Non-linear estimation is easy,'' \textit{Int. J. Model., Ident. Control}, vol. 4, no. 1, pp. 12--27, 2008.

\bibitem{ref9}
R. Ortega, A. Astolfi, G. Bastin, and H. Rodriguez, ``Stabilization of food-chain systems using a port-controlled Hamiltonian description,'' in \textit{Proc. Amer. Control Conf.}, Chicago, IL, USA,
2000, pp. 2245--2249.

\end{thebibliography}

\end{document}